\documentclass[14pt]{article}
\usepackage{amssymb,amsmath} 
\usepackage{amsthm,mathrsfs,mathtools}                                                
\usepackage{verbatim}    
\usepackage{graphicx}                                   
\usepackage{cite,float}    
\usepackage{algorithm,algpseudocode}     
\usepackage{subfig}
\usepackage{tikz, pgfplots, ifthen, xcolor, verbatim}
\usetikzlibrary{shapes,arrows, positioning, calc,fit, shapes.callouts}

\definecolor{pCrbColor1}{rgb}{0.0,0.5,0.0}%
\definecolor{pCrbColor2}{rgb}{0.9,0.0,0.0}%
\definecolor{pMaxLikelihoodColor}{rgb}{0.9,0.0,0.0}%
\definecolor{pMaxLikelihoodColor2}{rgb}{0.0,0.0,0.9}%

\definecolor{mycolor1}{rgb}{1.00000,0.00000,1.00000}%
\definecolor{mycolor2}{rgb}{0.00000,0.49804,0.00000}%

\newcommand{\pLineW}{0.6pt}
\newcommand{\pLineWS}{0.6}
\newcommand{\pMarkSize}{2}
\newcommand{\pFigW}{0.7}
\newcommand{\pFigH}{0.4}

\newcommand{\pTickSize}{\scriptsize}

\newcommand{\pTextSize}{\scriptsize}

\newcommand{\pSnr}{\pTextSize{$\rm{SNR (dB)} $}}
\newcommand{\pRmse}{\pTextSize{$\rm{RMSE (degree)} $}}
\newcommand{\pResolution}{\pTextSize{ {Source Resolution Percentage (\%)} }}
\newcommand{\pSamples}{\pTextSize{ { Number of Snapshots $(N)$ } }}
\newcommand{\pL}{\pTextSize{ { Number of Sources $(L)$ } }}
\newcommand{\pCrbUC}{\pTextSize{$\rm{CRB} $}}
\newcommand{\pMaxLikelihood}{\pTextSize{ {MLE} }}

\newcommand{\pSPICE}{\pTextSize{ {SPICE} }}

\newcommand{\pCrbAllWithLess}{\pTextSize{$\rm{CRB}$ for the case S1}}
\newcommand{\pCrbOneSubarrayCanIdentify}{\pTextSize{$\rm{CRB}$ for the case S2}}
\newcommand{\pCrbOneSubarray}{\pTextSize{$\rm{CRB}$ of the subarray with $M_{k}=3$ sensors}}

\newcommand{\caseCone}{S1}
\newcommand{\caseCtwo}{S2}

\newcommand{\rmse}{RMSE}

\newcommand{\substep}[1]{\State \textbf{\quad$\bullet$}}

\newcommand{\zeros}{\pmb{0}}

\newcommand{\eq}{=}
\newcommand{\mdots}{,\ldots,}
\newcommand{\expect}{\mathbb{E}}
\newcommand{\Id}[1]{\pmb{I}_{#1}}

\newcommand{\diag}{ {\rm{diag}}  }
\newcommand{\katri}{\circ}
\renewcommand{\vec}{ {\rm{vec} }}
\renewcommand{\det}[1]{| #1 |}

\newcommand{\mj}{\jmath}
\newcommand{\Floor}[1]{\lfloor #1 \rfloor}
\newcommand{\reals}{\mathbb{R}}
\newcommand{\ones}[1]{\pmb{1}_{#1}}

\newcommand{\GammaFunc}[2]{\Gamma_{#1}^{c}(#2)}

\newcommand{\Tr}[1]{ {\rm{tr}} \left(#1\right) }
\newcommand{\LogLH}[1]{\mathcal{L}(#1)}
\newcommand{\Fim}[1]{
\ifthenelse{\equal{#1}{}}
{{\rm{FIM}}}
{[{\rm{FIM}}]_{#1}}
}
\newcommand{\D}[2]{ \frac{d #1}{d #2}}

\newcommand{\algsectionheading}[1]{ \ifthenelse {\equal{#1}{}} {part} {Part}}

\newcommand{\snr}{\text{SNR}}

\makeatletter 
\newcommand*{\inlineequation}[2][]{%
  \begingroup
    \refstepcounter{equation}%
    \ifx\\#1\\%
    \else
      \label{#1}%
    \fi
    \relpenalty=10000 %
    \binoppenalty=10000 %
    \ensuremath{%
      #2%
    }%
    ~\@eqnnum
  \endgroup
}
\makeatother

\makeatletter
\newcommand{\stringcases}[3]{%
  \romannumeral
    \str@case{#1}#2{#1}{#3}\q@stop
}
\newcommand{\str@case}[3]{%
  \ifnum\pdf@strcmp{\unexpanded{#1}}{\unexpanded{#2}}=\z@
    \expandafter\@firstoftwo
  \else
    \expandafter\@secondoftwo
  \fi
    {\str@case@end{#3}}
    {\str@case{#1}}%
}
\newcommand{\str@case@end}{}
\long\def\str@case@end#1#2\q@stop{\z@#1}
\makeatother

\newcommand{\blkdiag}{{\rm blkdiag}}
\newcommand{\card}{{\rm card}}

\newcommand{\bl}{\pmb{b}}
\newcommand{\bls}{\mathcal{B}}
\newcommand{\eqMat}{\pmb{W}}

\newcommand{\Delay}{\tau}
\newcommand{\Delayk}{\Delay_k}

\newcommand{\ns}{\nsim}
\newcommand{\prop}{\mathcal{P}}

\newcommand{\M}{\breve{M}}

\newcommand{\RNK}{\rho}

\newcommand{\DOA}{\theta}
\newcommand{\DOAs}{\pmb{\DOA}}
\newcommand{\DOAsc}{\pmb{\nu}}
\newcommand{\gDOA}{\tilde{\DOAs}}

\newcommand{\MSR}{\pmb{x}}
\newcommand{\SRC}{\pmb{s}}
\newcommand{\NOS}{\pmb{n}}

\newcommand{\STR}{\pmb{A}}
\newcommand{\STRv}{\pmb{a}}
\newcommand{\STRkn}{\pmb{V}}
\newcommand{\STRknv}{\pmb{v}}

\newcommand{\PHASE}{\pmb{\Phi}}
\newcommand{\PHASEs}{\phi}
\newcommand{\bSTRkn}{\breve{\STRkn}}

\newcommand{\DISP}{\pmb{\zeta}}

\newcommand{\DISPl}{\pmb{\zeta}'}

\newcommand{\MCOV}{\pmb{R}}
\newcommand{\MCOVs}{\hat{\MCOV}}
\newcommand{\MCOVv}{\pmb{r}}
\newcommand{\MCOVvs}{\hat{\MCOVv}}

\newcommand{\SCOV}{\pmb{P}}

\newcommand{\SCOVd}{\pmb{\Lambda}}
\newcommand{\SCOVdv}{\pmb{\lambda}}
\newcommand{\SCOVds}{{\lambda}}
\newcommand{\SCOVo}{\pmb{F}}

\newcommand{\NCOV}{\sigma^2}

\newcommand{\gMCOV}{\tilde{\MCOV}}
\newcommand{\gSTRkn}{\tilde{\STRkn}}
\newcommand{\gSCOVdv}{\tilde{\SCOVdv}}
\newcommand{\gSCOVd}{\tilde{\SCOVd}}
\newcommand{\gSCOVds}{\tilde{\SCOVds}}
\newcommand{\gSCOVdvs}{\hat{\tilde{\SCOVdv}}}

\newcommand{\GRD}{\tilde{\DOAs}}
\newcommand{\pGRD}{\tilde{\DOA}}
\newcommand{\GRDl}{G}

\newcommand{\ID}{\pmb{I}}
\newcommand{\IDv}{\pmb{i}}

\newcommand{\MCovBig }{\breve{ \MCOV } }
\newcommand{\CrbPartOne}{\pmb{\Delta}_1}
\newcommand{\CrbPartTwo}{\pmb{\Delta}_2}
\newcommand{\AtHighSNRs}{|_{\SCOVds \gg \NCOV}}
\newcommand{\MCOVvi}{ \MCOVv\AtHighSNRs }
\newcommand{\MCOVi}{ \MCovBig\AtHighSNRs}
\newcommand{\MCOVii}{\overline{\pmb{V}}}

\newcommand{\CrbPartOnei}{\pmb{\Delta}_1\AtHighSNRs}
\newcommand{\CrbPartTwoi}{\pmb{\Delta}_2\AtHighSNRs}

\newcommand{\Crb}{{\rm{CRB}}}
\newcommand{\Crbi}{ { \rm{CRB} }_{\DOAs}^{-1}\AtHighSNRs }

\newcommand{\RNKi}{\overline{\RNK}}
\newcommand{\SCOVv}{\pmb{p}}
\newcommand{\corr}{\epsilon}
\newcommand{\MCORR}{\pmb{\Upsilon}}
\newcommand{\CrbPartOneic}{{{\pmb{u}}}}
\newcommand{\MCOVic}{\overline{\overline{\pmb{V}}}}

\newcommand{\SpiceW}{w}
\newcommand{\SpiceWNos}{\overline{w}}

\usepackage{titlesec} 

\titleformat{\section}{\large\bfseries}{\thesection.}{.5em}{}
\titlespacing*{\section}{0pt}{*3}{*2}
\titleformat{\subsection}{\normalfont\bfseries}{\thesubsection.}{.5em}{}
\titlespacing*{\subsection} {0pt}{*3}{*2}
\titleformat{\subsubsection}{\normalfont\bfseries}{\thesubsubsection.}{.5em}{}
\titlespacing*{\subsubsection} {0pt}{*3}{*2}

\theoremstyle{plain}
\newtheorem{thm}{Theorem}[] %

\newtheorem{crly}{Corollary}[]

\newtheorem{dfn} {Definition}[]

\begin{document}

\author{Wassim Suleiman, 
        Pouyan Parvazi, %
        Marius Pesavento,
        and~Abdelhak M. Zoubir%
}
\date{}  
\title{Non-Coherent Direction-of-Arrival Estimation Using Partly Calibrated Arrays}        

\maketitle

\begin{abstract}
In this paper, direction-of-arrival (DOA) estimation using 
non-coherent processing for partly calibrated arrays
composed of multiple subarrays is considered. 
The subarrays are assumed to compute locally the sample covariance matrices
of their measurements and communicate them to the
processing center.
A sufficient condition for 
the unique identifiability of the sources in
the aforementioned non-coherent processing scheme
is presented.
We prove that, under mild conditions, 
with the non-coherent system of subarrays,
it is possible 
to identify more sources
than identifiable by each individual subarray.
This property of non-coherent processing has not been investigated before.
We derive the Maximum Likelihood estimator (MLE)
for DOA estimation at the processing center using the sample covariance matrices received from the subarrays.
Moreover, the Cram\'er-Rao Bound (CRB) for 
our measurement model is derived 
and is used to assess the presented DOA estimators.
The behaviour of the CRB at high signal-to-noise ratio (SNR) is analyzed.
In contrast to coherent processing,
we prove that the CRB approaches zero at high SNR only if at least one subarray
can identify the sources individually.
\end{abstract}

\section{Introduction}

DOA estimation using sensor arrays 
plays a fundamental role in many applications such as radar, sonar
and seismic exploration \cite{VanTrees2002}. 
Centralized subspace-based DOA estimation algorithms such as MUSIC \cite{Schmidt1986},
root-MUSIC \cite{barabell1983improving}, MODE
\cite{Stoica1990}, and WSF \cite{Viberg1991}
exhibit the super-resolution property and are asymptotically efficient. 
These algorithms 
are applicable only when all sensor locations are known, i.e.,
the array is fully calibrated.
For partly calibrated arrays with unknown displacements between the subarrays, 
subspace-based algorithms, such as ESPRIT \cite{Roy1989}, RARE \cite{Pesavento2002}
and
algorithms proposed in \cite{Parvazi2011} and \cite{parvazi2011new}, can be applied. 
These algorithms perform coherent processing, i.e, 
they require the covariance matrix of the whole array including the inter-subarray covariance matrices.
Consequently, the subarrays are required to send their raw measurements to the processing center (PC)
which then computes the overall array covariance matrix.
Disadvantages encountered in coherent processing include the huge communication overhead at the subarrays
and the high computational load at the PC.

Since, non-coherent processing techniques are carried out using 
only the subarray covariance matrices \cite{stoica1995decentralized}, 
the largest available covariance lag in non-coherent processing
is the one corresponding to the subarray with the largest aperture, 
i.e., the subarray which possesses the largest inter-sensor distance.
Whereas, in coherent processing, the largest available covariance lag
corresponds to the whole array aperture which is larger than that of the individual subarrays.  
Thus, the DOA estimation performance of non-coherent processing 
is inferior to that of coherent-processing.
Nevertheless, non-coherent processing is preferred in large wireless sensor networks
since it offers a huge reduction in the communication overhead associated with  
communicating the raw subarray measurements to the PC as required in coherent processing.
The computational load associated with non-coherent processing is also 
much smaller than that of the coherent processing,
since only the small subarray covariance matrices are computed and not the large overall array covariance matrix.
Thus, non-coherent processing is more convenient for decentralized processing \cite{stoica1995decentralized}.
Moreover, the computation of the inter-subarray covariance matrices in coherent processing
requires synchronized subarrays, 
which is not always possible especially for widely separated subarrays \cite{stoica1995decentralized}.
Hence, in large arrays,
it is necessary to resort to non-coherent processing.
In such cases, 
the measurements of each subarray are processed coherently,
namely the subarray covariance matrices are computed locally at the subarrays 
and communicated to the PC.
Then, in the PC, 
non-coherent processing (using only local subarray covariance matrices) 
is carried out to achieve the DOA estimation task.

In \cite{wax1985decentralized,rieken2004generalizing},  
the MUSIC algorithm is generalized to non-coherent processing
where it is assumed that the subarrays locally estimate their noise subspaces
and send them to the PC.
In \cite{Soderstrom1992}, another version of the MUSIC algorithm for non-coherent processing is analyzed. 
In this algorithm, the subarrays send the locally estimated DOAs and their estimated
variances to the PC. 
A similar method which is robust against uncertainties in 
the statistical distribution of the noise 
is presented in \cite{lee1990robust}.
In \cite{stoica1995decentralized}, it is proposed to perform DOA estimation 
using the MODE algorithm individually in each subarray.
At the PC, the DOA estimates are optimally combined as in \cite{Soderstrom1992}.
In \cite{suleiman2014noncoherent}, the root-MUSIC algorithm \cite{barabell1983improving} is 
generalized for non-coherent processing
where the subarrays locally compute the root-MUSIC polynomial coefficients  
and communicate them to the PC.
Although the algorithms presented in 
\cite{wax1985decentralized,lee1990robust,Soderstrom1992,stoica1995decentralized,rieken2004generalizing,suleiman2014noncoherent} 
are designed for non-coherent processing, 
they all assume that each subarray can locally identify all the sources.
Our primary goal in this paper is to overcome this restricting assumption.
 
In \cite{sheinvald1999direction}, direction finding 
using fewer receivers than the number of sources is introduced.
Since only fewer receiver than the sources (and hence fewer than the sensors)
are available, it is impossible to sample the output of all the sensors simultaneously.
Thus, time varying processing is introduced where a different 
subset of the available sensors are sampled at each time period and
their measurement covariance matrix is computed.
The DOA estimation problem in this context  
can be considered as a non-coherent processing DOA estimation problem,
since the covariance matrices between different
sensor subsets are not available.
However, the authors of \cite{sheinvald1999direction} assume a fully calibrated array,
whereas this assumption is not made in our paper.
Moreover, the algorithms introduced in \cite{sheinvald1999direction}
perform an exhaustive search over the directions which is impractical 
when the number of sources is larger than two. 

In this paper, DOA estimation using non-coherent processing for partly calibrated arrays is considered.
We focus on the case where none of the subarrays is able to identify all the sources locally.
We present a bound on the maximum number of identifiable sources.
Using this bound, we show that for particular array geometries,
it is possible to identify more sources than each subarray can identify individually.
Thus, we achieve DOA estimation
in more general scenarios than considered in  
\cite{wax1985decentralized,lee1990robust,Soderstrom1992,stoica1995decentralized,rieken2004generalizing,suleiman2014noncoherent}.

Furthermore, 
two DOA estimation approaches are proposed:
1) the MLE
and 2) 
a computationally simpler DOA estimation approach based on sparse signal representation (SSR).
Moreover, the Cram\'er-Rao Bound (CRB)
for our measurement model is presented and analyzed.

We remark that the non-coherent processing based DOA estimation approaches 
considered in this paper and in
\cite{wax1985decentralized,lee1990robust,Soderstrom1992,stoica1995decentralized,rieken2004generalizing,suleiman2014noncoherent}
differs from that presented in \cite{kim2015non}.
Where, in \cite{kim2015non}, DOA estimation is achieved from magnitude only measurements.
Thus, the approach of \cite{kim2015non} introduces ambiguities in DOA estimation which 
have been resolved by assuming sources at known locations.
However, the approach of \cite{kim2015non} assumes less information about the structure of the subarrays
when compared to the approaches considered in this paper and in 
\cite{wax1985decentralized,lee1990robust,Soderstrom1992,stoica1995decentralized,rieken2004generalizing,suleiman2014noncoherent}.

The remainder of the paper is organized as follows.
In Section~\ref{sec:signal-model}, the signal model is introduced.
The case of uncorrelated sources is considered in Section~\ref{sec:uc}.
The model parameter identifiability is studied in Section~\ref{sec:id}.
The MLE and the CRB are derived
in Section~\ref{sec:mle} 
and in Section~\ref{sec:crb}, respectively.
DOA estimation based on the SSR approach
is proposed in Section~\ref{sec:ssr-uc}.
In Section~\ref{sec:extention},
the MLE and the CRB are
 extended to the case of correlated sources.
In Section~\ref{sec:simulation}, simulation results are presented.

In this paper,
lower-case bold symbols are used to denote vectors where
upper-case bold symbols denote matrices.
The transpose, complex conjugate, and the Hermitian operators
are denoted as $(\cdot)^T$, $(\cdot)^*$, and $(\cdot)^H$, respectively.
The symbols $\katri$ and $\otimes$ 
denote the Khatri-Rao and Kronecker products, respectively. 
The determinant and the trace of  a matrix are denoted as $\det{\cdot} $
and  $\Tr{\cdot}$, respectively.
The symbols $\Id{i}$, $\text{diag}(\cdot)$, $\text{blkdiag}(\cdot)$, $\vec (\cdot )$,
$[\pmb{A}]_{i,j}$, and $[\pmb{a}]_{i}$
denote
the identity matrix of size $i \times i$,
diagonal matrices, block diagonal matrices,
vectorization of a matrix,
the $(i,j)$th entry of a matrix, and the $i$th entry of a vector,
respectively.  
We write 
$\ones{i}$,
$\zeros_i$,
and
$\IDv_k$
to denote the vector of all ones of size $i$,
the vector of all zeros of size $i$,
and vectorization of the identity matrix of size 
equal to the number of sensors at the $k$th subarray,
respectively.
The expectation of random variables,
the floor operator, and
the magnitude and angle of complex variables
are denoted as 
$\expect (\cdot) $,
$\Floor{ \cdot }$, 
$|{\cdot}|$, and $\angle \cdot$, respectively,
where the imaginary unit is denoted as $\mj$.
The derivative of a function $f(x)$
with respect to a variable $x$ is denoted as
$\D{f(x)}{x}$.

\section{Signal Model}
\label{sec:signal-model}
We consider a planar partly calibrated array
composed of $K$ perfectly calibrated subarrays,
i.e., for each subarray the following assumptions hold:
\begin{enumerate}
  \item[A1] the narrow-band assumption, 
  \item[A2] the relative sensor locations are known,
  \item[A3] the sensors use a common sampling clock.
\end{enumerate}
However, these assumptions do not hold between the subarrays.
The $k$th subarray is comprised of $M_k$ sensors,
thus, the total number of sensors in the array is 
$M \eq \sum_{k=1}^K M_k$.
We define 
$\DISP_k \in \reals^2$
as the vector containing the unknown displacement
of the first sensor (reference) of the $k$th subarray 
and the reference sensor of the first subarray, thus,
$\DISP_1 =[0,0]^T$.
The considered known
 relative position of the $i$th sensor of the $k$th
subarray with respect to the first sensor of the $k$th subarray 
is denoted by $\DISPl_{k, i}$, for $i=1,\ldots,M_k$, and thus $\DISPl_{k, 1}=[0,0]^T$.

Signals of $L$ narrow-band far-field sources impinge onto the array
from directions $\DOAs \eq [\DOA_1, \ldots, \DOA_L]^T$.
The response of the $k$th subarray corresponding to a source at direction $\DOA$ is given by
\begin{equation}
\STRv_k (\DOA) \eq
\STRknv_k (\DOA) \PHASEs(\DOA, \DISP_k),
\label{eq:ak}
\end{equation}
where 
$\PHASEs(\DOA, \DISP_k) \eq \exp( \mj \frac{2 \pi}{\lambda_c}  \DISP_k^T \DOAsc(\DOA) ) $
is an unknown phase shift, 
$\lambda_c$ is the wavelength corresponding to the signal carrier frequency,
and $\DOAsc(\DOA)=[\sin(\DOA),\cos(\DOA)]^T$.
The vector $\STRknv_k(\DOA)$ is defined as
\begin{equation} 
\STRknv_k(\DOA) \eq [1, \exp(\mj \frac{2 \pi}{\lambda_c} \DISPl^T_{k,2} \DOAsc(\DOA) ), \ldots,
\exp(\mj \frac{2 \pi}{\lambda_c} \DISPl^T_{k, M_k} \DOAsc(\DOA) )]^T.
\label{eq:vk}
\end{equation}
In contrast to the phase shift $\PHASEs(\DOA, \DISP_k)$,
the vector $\STRknv_k (\DOA)$ is fully known as a function of $\DOA$.

The vector of the baseband signals received 
at the $k$th subarray is given by
\begin{equation}
\MSR_k(t-\Delayk) = \STR_k(\DOAs, \DISP_k) \SRC(t-\Delayk) + \NOS_k(t -\Delayk)
\label{eq:xk}
\end{equation}
where  $\Delayk$ is the sampling offset 
at the $k$th subarray and
$\NOS_k(t-\Delayk)$
is the vector containing the 
complex circular Gaussian sensor noise with 
zero-mean and variance $\NCOV$.
The vector
$\SRC(t-\Delayk)$
contains the complex circular Gaussian source signals
with zero-mean and covariance $\SCOV$.
We assume that the statistical properties of the sources
observed by different subarrays are identical,
i.e., 
\begin{equation}
\SCOV = \expect \big( \SRC(t-\Delay_k) \SRC^H(t-\Delay_k)   \big),
\label{eq:src-power}
\end{equation}
for $k=1 \mdots K$.
The steering matrix 
$\STR_k(\DOAs, \DISP_k) \eq
[\STRv_{k}(\DOA_1, \DISP_k),\ldots, \STRv_{k}(\DOA_L, \DISP_k)]^T
$ is written as
\begin{equation}
\STR_k (\DOAs, \DISP_k) = \STRkn_k(\DOAs) \PHASE_k(\DOAs, \DISP_k),
\label{eq:Ak}
\end{equation} 
where the matrix 
\begin{equation}
\STRkn_k(\DOAs)=[\STRknv_k(\DOA_1),\ldots,\STRknv_k(\DOA_L)]
\label{eq:Vk}
\end{equation}
depends only on the DOAs, whereas the diagonal matrix
\begin{equation}
\PHASE_k(\DOAs, \DISP_k)= \diag \big( \PHASEs_{k1} ,\ldots, \PHASEs_{kL} \big),
\label{eq:Phik}
\end{equation}
for $\PHASEs_{kl} = \PHASEs(\DOA_l, \DISP_k)$,
depends on the DOAs and the unknown displacements in $\DISP_k$.
In the following, 
the dependency on $\DOAs$ and $\DISP_k$ is dropped for notation convenience.

The true measurement covariance matrix 
of the $k$th subarray is written as
\begin{equation}
\MCOV_k \!= \!\expect \big( \MSR_k\!(t\!-\!\Delayk) \MSR_k^H\!(t\!-\!\Delayk) \big)
=\STRkn_k \PHASE_k \SCOV \PHASE_k^H \STRkn_k^H
 + \NCOV \ID_{M_k}, 
\label{eq:Rk}
\end{equation}
where the $M_k\times M_k$ identity matrix is denoted by $\ID_{M_k}$
and $\SCOV$ is defined in (\ref{eq:src-power}).
For the later use,
the source covariance matrix $\SCOV$ is partitioned as
\begin{equation}
\SCOV \eq \SCOVd + \SCOVo,
\label{eq:lambda-f}
\end{equation}
where the matrices $\SCOVd$ 
and $\SCOVo$ contain the diagonal and off-diagonal entries of 
the matrix $\SCOV$, respectively.
Denote the diagonal entries of the matrix
$\SCOVd$ as $\SCOVds_l$, for $l=1,\ldots,L$,
then $\SCOVds_l$ corresponds to the power of the $l$th source,
$\SCOVds_l > 0$.
We define
\begin{equation}
\SCOVdv = [\SCOVds_1,\ldots,\SCOVds_L]^T
\label{eq:lambda}
\end{equation}
to be the diagonal of the matrix $\SCOVd$.
The $(i,j)$th entry of matrix $\SCOVo$,
denoted as $[\SCOVo]_{i,j}$
corresponds to the correlation between  
the $i$th and $j$th sources.
The $i$th and $j$th sources are coherent or fully correlated when 
$\left|[\SCOVo]_{i,j}\right|  = \sqrt{\SCOVds_i \SCOVds_j}$.

The sample estimate of $\MCOV_k$ is computed using $N$ snapshots of the $k$th subarray output
as
\begin{equation}
\MCOVs_k =  \frac{1}{N} \sum_{t=1}^{N} \MSR_k(t)  \MSR_k^H (t), 
\label{eq:Rk-hat}
\end{equation}  
where without loss of generality, we assume that the same number of samples $N$
is available at all subarrays.

In this work, we assume that 
the subarrays send their locally estimated sample covariance matrices 
$\MCOVs_k$, for $k=1,\ldots,K$,
to the PC\footnote{%
This requires sending $M_k^2$ real numbers to the PC, instead of $2 N M_k$ in the case of sending raw  measurements.%
}, 
which carries out the DOA estimation algorithm.
This processing type is referred to as 
\emph{non-coherent processing} \cite{stoica1995decentralized},
since only the local subarray covariance matrices are available at the PC.
Compared to \emph{coherent processing}
where the sample estimate of the cross-subarrays covariance matrices,
i.e., $\expect[\MSR_k(t)\MSR_i^H(t)]$, for $i\neq k$, $i,k=,1\ldots,K$, 
are available at the PC\footnote{%
Note that in the model (\ref{eq:xk}),
the computation of $\expect[\MSR_k(t-\Delay_k)\MSR_i^H(t-\Delay_i)]$
when $k \neq i$ yields a covariance of zero if $|\Delay_k-\Delay_i|$
exceeds the coherence time of the signal waveforms such 
 that $\expect[\SRC(t-\Delay_k)\SRC^H(t-\Delay_i)]=\zeros_L\zeros_L^T$.}.
Which requires a synchronized subarray system,
i.e., $\Delayk=0$ for $k=1,\ldots,K$. 
We remark that: 
\begin{itemize}
  \item In non-coherent processing, 
  the resolution capability of the array is limited, compared to coherent processing, 
  since the largest available covariance lag corresponds to the largest subarray.
  Whereas, in coherent processing, 
  the largest available covariance lag corresponds to the
  array aperture.
  \item The non-coherent processing scheme is more suitable 
  for decentralized processing than the coherent processing one,
  since each subarray can act as a decentralized processing node 
  which computes the local covariance matrix of the subarray
  and sends it to the PC. Whereas, in coherent processing, 
  the computation of the cross-subarray covariance matrices requires
  either sending the raw measurement to the PC or the use of the averaging consensus (AC) protocol, i.e.,
  it involves a much larger communication overhead, see \cite{Scaglione2008,suleiman2015j2}.    
\end{itemize}

\section{DOA Estimation for Uncorrelated Sources}
\label{sec:uc}

In this section, 
we consider the special case of perfectly uncorrelated 
sources for which the structure of
the covariance matrix introduced in (\ref{eq:Rk}) 
can be  simplified.
We analyze the identifiability of our model 
and derive the CRB and the MLE.
Moreover, DOA estimation using SSR is presented.

Under the assumption of uncorrelated sources,
the source covariance matrix $\SCOV$ in (\ref{eq:src-power}) is diagonal, 
i.e., the entries of the cross-correlation matrix $\SCOVo$ in (\ref{eq:lambda-f}) 
are zeros 
and $\SCOV=\SCOVd$.
Since the matrix $\PHASE_k$ is also diagonal 
with unit amplitude entries we can write
\begin{equation}
\PHASE_k \SCOV \PHASE_k^H = \SCOV = \SCOVd.
\label{eq:Pk-uc} 
\end{equation}

Substituting (\ref{eq:Pk-uc}) in (\ref{eq:Rk}) 
yields
\begin{equation}
\MCOV_k =  \STRkn_k \SCOVd \STRkn_k^H  + \NCOV \ID_{M_k}.
\label{eq:Rk-uc}
\end{equation}

In \cite{graham1981kronecker}, 
the following matrix identity 
regarding the vectorization of the product of
three matrices $\pmb{M}_1, \pmb{M}_2, $ and $\pmb{M}_3$ of appropriate sizes
is proved:
\begin{equation}
\vec\big( \pmb{M}_1 \pmb{M}_2 \pmb{M}_3 \big) =
(\pmb{M}_3^T \otimes \pmb{M}_1) \vec(\pmb{M}_2).
\label{eq:vec} 
\end{equation}
Denote as $\MCOVv_k=\vec(\MCOV_k)$ the vectorization of the $k$th subarray 
measurement covariance matrix. 
Then, substituting (\ref{eq:Rk-uc}) and (\ref{eq:vec}) in $\MCOVv_k$ yields
\begin{equation} 
\MCOVv_k  = \big( \STRkn_k^* \otimes \STRkn_k \big) \vec(\SCOVd)
 + \NCOV \IDv_k,
\label{eq:rk-uc-m}
\end{equation}
where %
$\IDv_k \eq \vec(\ID_{M_k})$.
Since $\SCOVd$ is a diagonal matrix,
(\ref{eq:rk-uc-m}) is further reduced to
\begin{equation}
\MCOVv_k  =  \bSTRkn_k \SCOVdv  + \NCOV \IDv_k,
\label{eq:rk-uc}
\end{equation} 
where the vector $\SCOVdv$ is defined in (\ref{eq:lambda}) and
the $M_k^2 \times L$ matrix 
\begin{equation}
\bSTRkn_k \eq \big( \STRkn_k^* \katri \STRkn_k \big)
\label{eq:bVk}
\end{equation}
contains the  
columns of the matrix $\big( \STRkn_k^* \otimes \STRkn_k \big)$
corresponding to the diagonal of $\SCOVd$.
The matrix $\bSTRkn_k$ is referred as the co-subarray manifold\footnote{%
The expression co-array manifold have been used in \cite{Abramovich1999Resolving}
in the context of nonuniform linear antenna arrays 
to denote the Katri-Rao product of the conjugate array response with itself.}.
We define the concatenation
of all vectorized measurement covariance matrices as
\begin{equation}
\MCOVv = [\MCOVv^T_1,\ldots,\MCOVv_K^T]^T,
\label{eq:r}
\end{equation}
where $\MCOVv$ is of size $\M \eq \sum_{k=1}^KM_k^2$.
By substituting (\ref{eq:rk-uc}) in (\ref{eq:r}),
the vector $\MCOVv$ becomes
\begin{equation}
\MCOVv  =  \bSTRkn \SCOVdv  + \NCOV \IDv,
\label{eq:r-uc}
\end{equation} 
where 
\begin{equation}
\bSTRkn \eq [\bSTRkn_1^T,\ldots,\bSTRkn_K^T]^T
\label{eq:bV}
\end{equation}
is the co-array manifold and
\begin{equation}
\IDv \eq [\IDv_1^T,\ldots,\IDv_K^T]^T.
\label{eq:i}
\end{equation}
We denote as $\MCOVvs$ and $\MCOVvs_k$, for $k=1,\ldots,K$,
the sample estimate of $\MCOVv$ and $\MCOVv_k$, respectively,
which are obtained from the sample covariance matrix in (\ref{eq:Rk-hat}). 

\subsection{Identifiability} 
\label{sec:id}
In this subsection, 
we first 
revise the condition of parameter identifiability as introduced in
\cite{hochwald1996identifiability},
then we present a sufficient condition on 
the maximum number of identifiable (uncorrelated) sources.

Let $\DOAs'=[\DOA_1',\ldots,\DOA_L']^T$ and $\DOAs''=[\DOA_1'',\ldots,\DOA_L'']^T$ 
denote two vectors each of them containing $L$ 
pairwise-different DOAs.
By pairwise-different DOA vector $\DOAs'$ we mean that 
$\DOA_i' \neq \DOA_j'$ for $i \neq j$ and $i,j=1,\ldots,L$.
Then, we write $\DOAs' \ns \DOAs''$ if there exist an index $i \leq L$
where for all $j \leq L$, $\DOA_i' \neq \DOA_j''$.
In other words, at least one entry of $\DOAs'$
is not equal to any entry of $\DOAs''$.
In the following, we present the definition of identifiability \cite{hochwald1996identifiability}. 
\begin{dfn}[Identifiability]
\label{dfn:id}
In the noise free case, $L$ sources 
with DOAs $\DOAs$ and powers $\SCOVdv$
are uniquely identifiable if 
\begin{equation}
\bSTRkn(\DOAs) \SCOVdv \neq \bSTRkn(\DOAs') \SCOVdv',
\label{eq:id} 
\end{equation}
for any vector with positive entries $\SCOVdv'$
and for any pairwise-different DOA vector $\DOAs'$,
where $\DOAs \ns \DOAs'$.%
\end{dfn}%
Note that in the noise free case,
the product $\bSTRkn(\DOAs) \SCOVdv$
consist in the vectorized measurement covariances,
i.e., $\MCOVv = \bSTRkn(\DOAs) \SCOVdv$.
Let $F(\MSR(t)|\DOAs)$
denotes the distribution of the array measurements
for a particular source directions $\DOAs$.
Since the subarray measurements follows a zero mean Gaussian distribution 
with (vectorized) covariances $\MCOVv$,
Definition~\ref{dfn:id} implies that,
the direction of the sources are uniquely identifiable 
if two parameter vectors $\DOAs$ and $\DOAs'$, where $\DOAs \ns \DOAs'$,
yield different measurement distributions,
i.e., $F(\MSR(t)|\DOAs) \neq F(\MSR(t)|\DOAs')$ for
$\DOAs \ns \DOAs'$ \cite{hochwald1996identifiability}.

Let $\RNK$ denotes the Kruskal rank \cite{stegeman2007kruskal,kruskal1977three} of the co-array manifold matrix $\bSTRkn$,
i.e., $\RNK$ is the largest integer such that
the columns of the matrix $\bSTRkn([\DOA_1 \mdots \DOA_{\RNK}]^T)$
are linearly independent for any 
vector $[\DOA_1 \mdots \DOA_{\RNK}]^T$ with pairwise different DOAs.
Based on $\RNK$, the following theorem
provides a sufficient condition for the unique identifiability of $L$ sources.
\begin{thm}[Sufficient condition for identifiability]
\label{thm:id}
The $L$ DOAs $\DOAs$ can be uniquely identified from covariances 
$\MCOVv = \bSTRkn \SCOVdv$ provided that 
\begin{equation}
L \leq  \Floor{\frac{ \RNK }{2} },
\label{eq:id-nc} 
\end{equation}  
where $\RNK$ is the Kruskal rank
of the co-array manifold $\bSTRkn$.
\end{thm}
\begin{proof}
See Appendix \ref{apx:thm-id}.
\end{proof}
Denote by $\bl_{k,i,j}$ the ($i,j$)th covariance lag of the $k$th subarray,
i.e., $\bl_{k,i,j} = \DISPl_{k,j}-\DISPl_{k,i}$
and let $\bls_k$ denotes the set of all different covariance lags of the  $k$th subarray,
i.e.,
\begin{equation}
\bls_k=\{\bl_{k,i,j}, i,j=1 \mdots M_k\}.
\end{equation} 
Further, let $\bls$ denotes the set of different covariance lags of the whole array,
i.e.,
\begin{equation}
\bls= \bigcup_{k=1}^K \bls_k.
\end{equation}
Then the Kruskal rank $\RNK$  of the matrix $\bSTRkn$
is bounded by the number of covariance lags in the set $\bls$.
This observation yields the following result.
\begin{crly}
\label{crly:id-1}
The number of sources which can be uniquely identified from covariances 
$\MCOVv$ is smaller than $\Floor{ \card( \bls )/2 }$,
where $\card( \bls )$ is the cardinality of the set $\bls$.  
\end{crly}
Corollary~\ref{crly:id-1} 
implies that the number of uniquely identifiable sources
using non-coherent processing can be increased by 
designing the subarrays with different covariance lags.
Note that if all subarrays admit the same covariance lags, 
e.g., if the subarrays are identical,
then the number of uniquely identifiable sources by the whole array is equal 
to the number identifiable by one individual subarray. 
The following example provides further insight.

\subsubsection*{Example}
Consider an array composed of $K=3$ identically oriented linear subarrays
where the $k$th subarray includes $M_k=2$ sensors.
The relative positions between the successive sensors in the
subarrays are assumed to be 
$d_1=1$, $d_2=2$ and $d_3=3$ half-wavelength, respectively,
see Fig.~\ref{fig:example1}.
For coherent processing the maximum number 
of identifiable sources using this array is 
$M-K=3$ (see \cite{Pesavento2002}).
Note that coherent processing scenario 
represents an upper bound on the number of uniquely identifiable sources 
using non-coherent processing,
since more covariance lags are available for coherent processing, 
namely, the covariance lags corresponding to the relative position 
of two sensors belonging to different subarrays.
Thus, $L\leq3$ is a necessary condition for identifying the sources using non-coherent processing.
In the following, based on Theorem~\ref{thm:id},
we show that $L\leq3$ is a sufficient condition for identifying the sources
in the considered array example.

The subarray steering vectors in (\ref{eq:vk}) are reduced to
$\STRknv_k(\DOA) = [1, e^{\mj d_k \pi \sin \DOA}]^T$,
for $k=1 \mdots 3$, in this example. 
Thus, the matrix
$\bSTRkn$ has the same rank as the matrix
\begin{equation}
\eqMat \eq
\left(\begin{matrix}
e^{-3 \jmath \pi \sin \DOA_1} &  \cdots & e^{-3 \jmath \pi \sin \DOA_L} \\
e^{-2 \jmath \pi \sin \DOA_1} &  \cdots & e^{-2 \jmath \pi \sin \DOA_L} \\
\vdots &  \vdots & \vdots \\
e^{3 \jmath \pi \sin \DOA_1} &  \cdots & e^{3 \jmath \pi \sin \DOA_L} \\
\end{matrix} \right),
\end{equation}
where we only rearranged and deleted duplicated rows from $\bSTRkn$
to get $\eqMat$.
The matrix $\eqMat$ is a Vandermonde matrix with $7$ rows.
Consequently, $\RNK=7$ and $\Floor{\frac{\RNK}{2}} = 3$, i.e., up to 
$L=3$ sources can be identified assuming non-coherent processing in this example.
Thus, regarding identifiability non-coherent processing
is equivalent to coherent processing in  this scenario. 
Moreover, observe that where each subarray is able to identify one source locally (since each subarray  consists of 2 sensors \cite{wax1989unique}),
using non-coherent processing, the number of identifiable sources
is increased up to $L=3$ sources. 
This increase results from the fact 
that the three subarrays have different covariance lags.

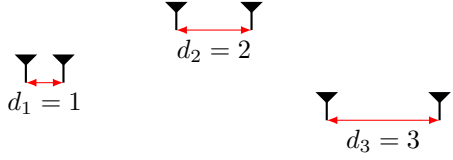
\begin{figure}
\centering
\newcommand{\Antenna}[3]
{
\coordinate [] (A) at (#1,     		#2);
\coordinate [] (C) at (#1,     		#2+#3*0.25*3); 
\coordinate [] (D) at (#1-#3*0.5, 	#2+#3*0.25*5);
\coordinate [] (E) at (#1+#3*0.5,	#2+#3*0.25*5);
\draw[thick] (A)--($(C)+(0,0.02)$);
\fill (C)--(E)--(D)--cycle ;
}

\begin{tikzpicture}
\def\AntennaScale{0.3}

\foreach \x/\y/\d/\k in {0/0/0.5/1, 2/0.7/1.0/2, 4/-0.5/1.5/3}{
	\Antenna{\x}{\y}{\AntennaScale}
	\Antenna{\x+\d}{\y}{\AntennaScale}
	\draw[latex-latex, red] (\x, \y) -- (\x+\d, \y) node[below, pos=0.5, text=black] {$d_{\k}=\k$}; 
}

\end{tikzpicture}
\caption{Array composed of $K=3$ subarrays.}
\label{fig:example1}
\end{figure}

\subsection{Maximum Likelihood Estimator}
\label{sec:mle}
In this section, the MLE for DOA estimation using non-coherent processing is derived considering uncorrelated sources.

In the scenario considered in this work,
the PC receives the sample covariance matrices from the subarrays.
These matrices follow a Wishart distribution
\cite[p.~49]{schreier2010statistical}
with probability density function (pdf) 
\begin{equation}
 \prop( \MCOVs_k ) = 
   \frac{ \det{  N \MCOVs_k }^{N-M_k}}
   {\GammaFunc{M_k}{N} \det{ \MCOV_k } ^{N} }    
  \exp\left( - N\Tr{ \MCOV_k^{-1} \MCOVs_k } \right) 
\label{eq:wishart}
\end{equation}
where
$\GammaFunc{M_k}{N} = \pi^{M_k(M_k-1)/2} \prod_{i=1}^{M_k}\prod_{j=1}^{N-i} j$
and $\MCOV_k$ is given in (\ref{eq:Rk-uc}).
Ignoring the constant term in (\ref{eq:wishart}),
the negative log-likelihood function is written as 
\begin{equation}
 \LogLH{\MCOV_1 \mdots \MCOV_K} = 
 \sum_{k=1}^{K} N \left( \log \det{ \MCOV_k } + \Tr{ \MCOV_k^{-1} \MCOVs_k } \right).
 \label{eq:mle-uc}
\end{equation}
The function $\LogLH{\MCOV_1 \mdots \MCOV_K}$ 
is valid under the assumption of correlated sources as well as uncorrelated sources.
Where only the structure of the measurement covariance matrices    
$\MCOV_1 \mdots \MCOV_K$ depends on the source correlations.
For uncorrelated sources the measurement covariance matrix of the $k$th subarray
$\MCOV_k$
reduces to (\ref{eq:Rk-uc}), i.e., 
$\MCOV_k$ depends on the DOAs $\DOAs$, the source powers $\SCOVdv$, and the noise variance $\NCOV$.
Thus, the DOAs, the power of the sources, and the noise variance are estimated 
by solving the minimization problem
\begin{equation}
\begin{aligned} 
& \underset{\DOAs, \SCOVdv, \NCOV}{\min}  \LogLH{ \DOAs , \SCOVdv, \NCOV} \\
&{\rm s.t.} \quad  \SCOVdv > \zeros_{L}, \\
& \quad\quad  \NCOV > 0.
\end{aligned}
 \label{eq:doa-uc}
\end{equation} 
The function $\LogLH{ \DOAs , \SCOVdv, \NCOV}$  in (\ref{eq:doa-uc}) 
is nonconvex \cite{boyd2004convex}.
Therefore, starting from a feasible point, 
a local solution (local minimum) for (\ref{eq:doa-uc}) can be computed,
e.g., using the gradient descent method \cite{boyd2004convex}.

\subsection{The Cram\'er-Rao Bound (CRB)}
\label{sec:crb}
In \cite{stoica1995decentralized},
an expression for the CRB using non-coherent processing is derived
under the assumption $M_k > L$, for $k=1 \mdots K$. 
Note that when $M_k < L$, the Fisher information matrix (FIM)  
corresponding to the $k$th subarray,
denoted by $\Fim{}_k$, is rank deficient. 
Therefore  the expression of \cite{stoica1995decentralized}
is no longer valid.    
The FIM matrix for the non-coherent processing scenario
\begin{equation}
\Fim{}  = \sum_{k=1}^K \Fim{}_k
\label{eq:fim-overall} 
\end{equation}
is used to compute the CRB.   
Using (\ref{eq:fim-overall})
and following the steps of \cite{sheinvald1999direction,sheinvald1997achievable}, 
the CRB corresponding to the direction parameters $\DOAs$  
can be written as 
\begin{equation}
\Crb_{\DOAs} = \left(  
\CrbPartOne^H \left(
\MCovBig - \MCovBig \CrbPartTwo 
  \left(\CrbPartTwo^H \MCovBig \CrbPartTwo \right)^{-1} 
\CrbPartTwo^H \MCovBig
\right) \CrbPartOne \right)^{-1}.
\label{eq:crb-uc}
\end{equation}
where
\begin{equation}
\CrbPartOne \eq [\D{\MCOVv} {\DOAs^T}], \quad
\CrbPartTwo \eq [ \D{\MCOVv} {\SCOVdv^T}, \D{\MCOVv} {\NCOV}],
\label{eq:delta}
\end{equation}
are the matrices which represent the derivatives of $\MCOVv$ with respect to 
$\DOAs$, $\SCOVdv$, and $\NCOV$, respectively, 
\begin{equation}
\MCovBig  \eq \blkdiag \left( \MCovBig_1 ,\ldots, \MCovBig_K  \right),
\label{eq:BIGR-11}
\end{equation}
and $\MCovBig_k \eq N (\MCOV_k^{-T} \otimes  \MCOV_k^{-1} )$.
In the sequel,  we demonstrate the behaviour of the CRB at high SNR by simulation
and we analyze this behaviour.
\begin{figure}

\definecolor{mycolor1}{rgb}{0.00000,0.49804,0.00000}%
\begin{tikzpicture}

\begin{axis}[%
width=\pFigW \columnwidth,
height=\pFigH  \columnwidth,
at={(1.352in,0.855556in)},
scale only axis,
separate axis lines,
every outer x axis line/.append style={black},
every x tick label/.append style={font=\color{black}},
xmin=-30.0000,
xmax=51.0000,
xmajorgrids,
xlabel={\pSnr},
every outer y axis line/.append style={black},
every y tick label/.append style={font=\color{black}},
ymode=log,
ymin=0.0010,
ymax=10000.0000,
yminorticks=true,
ymajorgrids,
ylabel={\pRmse},
legend style={legend cell align=left,align=left,fill=white}
]

\addplot [color=pCrbColor1,solid,line width=\pLineWS pt
, smooth
]
  table[row sep=crcr]{%
-30.0000	263.0516\\
-26.0000	105.0385\\
-22.0000	42.1318\\
-18.0000	17.0872\\
-14.0000	7.1142\\
-10.0000	3.1384\\
-6.0000	1.5458\\
-2.0000	0.8981\\
2.0000	0.6263\\
6.0000	0.5067\\
10.0000	0.4512\\
14.0000	0.4245\\
18.0000	0.4116\\
22.0000	0.4055\\
26.0000	0.4028\\
30.0000	0.4017\\
34.0000	0.4012\\
38.0000	0.4010\\
42.0000	0.4009\\
46.0000	0.4009\\
50.0000	0.4009\\
};
\addlegendentry{\pCrbAllWithLess};

\addplot [color=blue,dashed
 ,mark=o,mark options={line width = 0.8 pt, solid}
 ,mark repeat={4}, mark phase = 3, mark size=2.5 pt
]
  table[row sep=crcr]{%
-30.0000	256.5259\\
-26.0000	102.4398\\
-22.0000	41.0965\\
-18.0000	16.6739\\
-14.0000	6.9480\\
-10.0000	3.0697\\
-6.0000	1.5146\\
-2.0000	0.8808\\
2.0000	0.6126\\
6.0000	0.4903\\
10.0000	0.4237\\
14.0000	0.3708\\
18.0000	0.3116\\
22.0000	0.2446\\
26.0000	0.1809\\
30.0000	0.1279\\
34.0000	0.0868\\
38.0000	0.0569\\
42.0000	0.0365\\
46.0000	0.0232\\
50.0000	0.0147\\
};
\addlegendentry{\pCrbOneSubarrayCanIdentify};

\addplot [color=pCrbColor2,dashed,line width=\pLineW
 ,mark=+,mark options={line width = \pLineW, solid}
 ,mark repeat={4}, mark phase = 1, mark size=\pMarkSize
, smooth
]
  table[row sep=crcr]{%
-30.0000	4833.3076\\
-26.0000	1925.5257\\
-22.0000	767.9152\\
-18.0000	307.0580\\
-14.0000	123.5772\\
-10.0000	50.5075\\
-6.0000	21.3618\\
-2.0000	9.6453\\
2.0000	4.7947\\
6.0000	2.6317\\
10.0000	1.5506\\
14.0000	0.9493\\
18.0000	0.5915\\
22.0000	0.3714\\
26.0000	0.2338\\
30.0000	0.1474\\
34.0000	0.0930\\
38.0000	0.0587\\
42.0000	0.0370\\
46.0000	0.0234\\
50.0000	0.0147\\
};
\addlegendentry{\pCrbOneSubarray};

\end{axis}
\end{tikzpicture}%
\caption{The CRB for the cases where 1) none of the subarrays are able 
to identify the sources individually
2) one  subarray can identify the sources.
Also the CRB for the subarray with 3 sensors in case 2) is shown.}
\label{fig:crb} 
\end{figure}
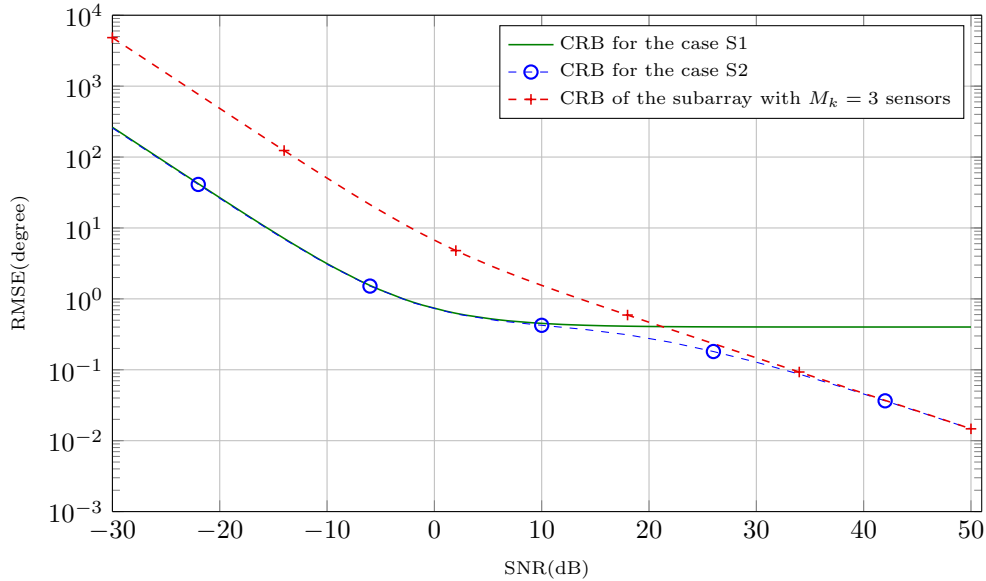 

Consider the following two scenarios:
\begin{enumerate}
  \item[{\caseCone}] $M_1 = \cdots = M_{K} \leq L$, i.e., the FIM for each individual subarray is not invertible,
  whereas the overall FIM, defined in (\ref{eq:fim-overall}), is invertible. 
  \item[{\caseCtwo}] $M_{1} > L$ and $M_k \leq L$, for $k=2 \mdots K$, i.e., 
  the FIM of the first subarray $\Fim{}_1$ is invertible
whereas the FIM of the remaining subarrays, i.e.,  
  $\Fim{}_k$, for $k=2,\ldots,K$ are not invertible.
\end{enumerate}
In Fig.~\ref{fig:crb}, we display the CRB for $K = 12$ subarrays and 
$L = 2$ uncorrelated equal-power sources 
for two array configurations which represent the aforementioned scenarios {\caseCone} and {\caseCtwo}\footnote{%
For the details on the array geometry parameters 
please refer the array setup described in Section~\ref{sec:simulation}. }:
\begin{enumerate}
  \item $M_1 = \cdots = M_{K} = 2 = L$, which represents {\caseCone}. 
  \item $M_{1}=3 > L$ and $M_2 = \cdots = M_{K-1} = 2=L$, which represents {\caseCtwo}.
\end{enumerate}
Moreover, 
in Fig.~\ref{fig:crb},
we display 
the CRB of the first subarray with $M_{1}=3$ sensors.
It can be observed from Fig.~\ref{fig:crb}
that in the scenario {\caseCone}, the CRB does not approach zero as the SNR approaches infinity
rather it remains unchanged at high SNR (in Fig.~\ref{fig:crb},
the CRB remains almost unchanged for SNR above $15$ dB).  
In the scenario {\caseCtwo}, the CRB is almost identical to that of the scenario {\caseCone}
when the SNR is less than $15$ dB.
However, it continues to decrease for SNR larger than $15$ dB 
and the performance at high SNR in this case
is determined by the performance of the first subarray. 
Thus, at high SNR, DOA estimation
can be performed using only those subarrays which are able to
identify and estimate the DOAs individually, if such subarrays exist.
In \cite{lee1990robust}, the authors
suggested to include only subarrays
which can individually identify all the sources
in the DOA estimation algorithm. 
This approach is justified at high SNR, 
however, at low SNR using all the subarray yields 
the better estimation performance, 
as demonstrated by the CRB in Fig.~\ref{fig:crb}.  

In the following, we analyze the aforementioned behaviour of the CRB
at high SNR in the two scenarios {\caseCone} and {\caseCtwo}.
Thus, we consider $L$ uncorrelated equally-powered sources in the high SNR region,
i.e., $\SCOVds_1 = \cdots = \SCOVds_L = \SCOVds$, 
where $\SCOVds_1 \mdots \SCOVds_L$ are the power of the sources whose directions 
are denoted by $\DOA_1 \mdots \DOA_L$, respectively, and 
$\SCOVds \gg \NCOV$, refer to (\ref{eq:lambda}).
Let $\MCOVvi \eq  \ 
\SCOVds\left( \bSTRkn \ones{L}  + \frac{\NCOV}{\SCOVds} \IDv \right) \AtHighSNRs \approx \SCOVds \bSTRkn \ones{L}$
denotes the high SNR approximation of the vectorized
covariance matrices.
Consequently, the derivative matrices  $\CrbPartOne$ and $\CrbPartTwo$
in (\ref{eq:delta})
reduce to  
$\CrbPartOnei = \SCOVds [\D{(\bSTRkn \ones{L})}{\DOAs^T} ]$ 
and 
$\CrbPartTwoi = [\bSTRkn \ones{L}, \IDv]$.
Similarly, we denote 
$\MCovBig$ at high SNR by
$\MCOVi  \approx \SCOVds^{-2} N \MCOVii$,
where  
$\MCOVii = \blkdiag \left( \MCOVii_1 ,\ldots, \MCOVii_K  \right)$
and $\MCOVii_k \eq (\STRkn_k \STRkn_k^H)^{-T} \otimes  (\STRkn_k \STRkn_k^H)_k^{-1}$.
Substituting $\CrbPartOnei$, $\CrbPartTwoi$, and $\MCOVi$ in (\ref{eq:crb-uc})
the CRB in the high SNR region reduces to
\begin{equation}
\Crbi  \approx  N
[\D{(\bSTRkn \ones{L})}{\DOAs^T} ]^H
\;\MCOVii\;
[\D{(\bSTRkn \ones{L})}{\DOAs^T} ].
\label{eq:crb-infty} 
\end{equation}
Interestingly,
we observe from (\ref{eq:crb-infty}) that
 at high SNR,
 the expression for  $\Crbi$  depend neither on $\SCOVds$
 nor on $\NCOV$ 
 but only on the DOAs $\DOA_1\mdots\DOA_L$.
Next, let us consider how the expression for  $\Crbi$ 
changes in the two scenarios {\caseCone} and {\caseCtwo}.
Let $\RNKi_k$ denote the rank of the matrix
$\MCOVii_k^{-1}$. 
Since the rank of the Kronecker product is the product of the ranks 
 of its operand matrices \cite{abadir2005matrix},
the rank $\RNKi_k$, for $k=1,\ldots,K$, takes the value $\RNKi_k=M_k^2$
in both scenarios {\caseCone} and {\caseCtwo}\footnote{%
 Using the well-known inversion identity
 $(\pmb{A} \otimes \pmb{B})^{-1}=\pmb{A}^{-1} \otimes \pmb{B}^{-1}$ 
 \cite{abadir2005matrix}.}.
Thus, the following behaviour of the block diagonal matrix $\MCOVii^{-1}$
is observable:  
\begin{itemize}
  \item In the scenario {\caseCone}, $\MCOVii^{-1}$ is full rank.
  \item In the scenario {\caseCtwo},  $\MCOVii^{-1}$ is rank deficient.
  More precisely, the first block of $\MCOVii^{-1}$, 
  which corresponds to the first subarray is rank deficient.
\end{itemize}
Consequently,
in the scenario {\caseCone},
the matrix $\MCOVii$ %
has finite entries (and eigenvalues) leading to a finite non-zero CRB.
Whereas, in the scenario {\caseCtwo},
the matrix $\MCOVii$  has infinitely large eigenvalues
which asymptotically drive the CRB to zero). 
Moreover, in the scenario {\caseCtwo}, $\Crbi$ in (\ref{eq:crb-infty}) can be approximated by 
$
 N
[\D{(\bSTRkn_1 \ones{L})}{\DOAs^T} ]^H
\;\MCOVii_1\;
[\D{(\bSTRkn_1 \ones{L})}{\DOAs^T}]
$
since the entries of $\MCOVii_2,\ldots,\MCOVii_K$ are negligible compared to the entries of $\MCOVii_1$.
Which means that in the scenario {\caseCtwo},  at high SNR,
the CRB of the whole array can be approximated by the CRB of the first subarray.
We remark that a behaviour of the CRB similar to that of scenario {\caseCone} 
at high SNR has been 
observed in \cite{sheinvald1999direction,sheinvald1997achievable}
for DOA estimation using fewer receivers
and 
``it is shown to be typical  in scenarios where a signal  subspace  is  nonexistent''.
However, in \cite{sheinvald1999direction,sheinvald1997achievable} 
the scenario {\caseCtwo} has not been considered.
Moreover, in \cite[Fig.~1]{abramovich1998positive},
a similar behaviour to the scenario {\caseCone} is observed in DOA estimation using
fully augmentable sparse linear arrays when the number of sources is larger than the number of 
the sensors in the array but smaller than the available covariance lags.

Regarding the number of samples $N$, we point out that 
the CRB approaches zero in both scenarios {\caseCone} and {\caseCtwo} when $N$ approaches infinity,
as it can be observed from (\ref{eq:crb-infty}).

\subsection{DOA Estimation Using Sparse Signal Representation}
\label{sec:ssr-uc}
Sparse signal representation (SSR)  
\cite{tibshirani1996regression,donoho2006stable,donoho2001uncertainty,donoho2008fast,candes2008introduction}
has recently attracted much attention in DOA estimation applications,
see \cite{malioutov2005sparse,hyder2010direction,steffens2014direction,yin2011direction,atashbar2012direction}.
One important advantage of 
SSR  is that it performs well in the low sample size regime.
Furthermore, using the norm $\ell_1$ relaxation the SSR can be cast as a convex optimization problem. 
So far, the focus of DOA estimation using SSR 
has been in the context of coherent processing  
\cite{malioutov2005sparse,hyder2010direction,steffens2014direction,yin2011direction,atashbar2012direction},
however, to the best of our knowledge the SSR approach 
has not yet been applied for non-coherent processing based DOA estimation.
In this section, 
we formulate the DOA estimation problem in the case of uncorrelated sources as a SSR problem,
which can be solved using convex optimization algorithms, see \cite{boyd2004convex,cvx}.

For coherent processing using fully calibrated array,
covariance based SSR approaches for deterministic and stochastic source models are introduced in
\cite{steffens2016compact} and \cite{stoica2011spice}, respectively.  
Since a stochastic source model is assumed in this paper, 
we extend the approach of \cite{stoica2011spice},
referred to as SParse Iterative Covariance-based approach (SPICE),
to non-coherent processing using partly calibrated arrays\footnote{%
The extension of \cite{steffens2016compact} to non-coherent processing using partly calibrated arrays is similar to that
of \cite{stoica2011spice}.
}.

Let $\GRD$ be the vector of length $\GRDl$ obtained by sampling the
field-of-view in $\GRDl \gg L$  angular directions 
\begin{equation}
\GRD = [\pGRD_1 \mdots \pGRD_G]^T.
\label{eq:grid}
\end{equation}
Then, the SPICE optimization problem \cite[Equation~(20)]{stoica2011spice} 
for the considered non-coherent processing scenario is written as 
\begin{subequations}
\begin{equation}
\begin{aligned}
& \underset{\gSCOVdv, \NCOV}{\min} \sum_{k=1}^K \Tr{\gMCOV^{-1}_k \MCOVs_k} \\
\end{aligned}
\end{equation}
\begin{equation}
\begin{aligned}
&{\rm s.t.} \quad  \gSCOVdv \geq \zeros_{G}, \;  \NCOV \geq 0, \\
\end{aligned}
\end{equation}
\begin{equation}
\begin{aligned}
& \quad\quad \sum_{g=1}^G \SpiceW_g  \gSCOVds_g + \SpiceWNos \NCOV= 1
\end{aligned}
\label{eq:spice-sparse}
\end{equation}%
\label{eq:spice}%
\end{subequations}%
where
$\gMCOV_k = \gSTRkn_k  \gSCOVd \gSTRkn^H+\NCOV\ID_{M_k}$ and
the $\M_k \times \GRDl$
overcomplete dictionary
$\gSTRkn_k$ is defined as 
\begin{equation}
\gSTRkn_k \eq [\STRknv_k(\gDOA_1),\ldots,\STRknv_k(\gDOA_{\GRDl})].
\label{eq:gSTRkn} 
\end{equation}
The diagonal matrix $\gSCOVd$ is a sparse matrix whose diagonal elements, denoted as $\gSCOVdv$,
correspond to the powers of the sources at directions $\gDOA$.
The weights in (\ref{eq:spice-sparse}) are defined as
\begin{equation}
\SpiceW_g = \frac{1}{M} \sum_{k=1}^K \STRknv^H_k(\gDOA_g) \MCOVs^{-1}_k \STRknv_k(\gDOA_g),
\end{equation}
and
\begin{equation}
\SpiceWNos = \frac{1}{M} \sum_{k=1}^K \Tr{\MCOVs^{-1}_k}.
\end{equation}
In \cite{stoica2011spice}, it has been pointed out that the constraint (\ref{eq:spice-sparse}) is
a weighted $\ell_1$ norm and thus is expected to induce sparsity. 
Note that in contrast to other $\ell_1$ norm based DOA estimation approaches, %
the SPICE approach does not require the configuration of a sparsity regularization parameter. 
Problem (\ref{eq:spice}) is positive semi-definite \cite{stoica2011spice} thus can be solved using,
e.g., cvx \cite{cvx}\footnote{%
Problem (\ref{eq:spice}) can be cast as second order cone program (SOCP)
and it can be extended to the case where the sensor noise variance
are not identical at all sensors, see \cite{stoica2011spice}.
}.

Note that using SSR, the DOA estimation problem is reduced to
the identification of the non-zero elements
in the estimated sparse vector $\gSCOVdvs$.
These non-zero elements are referred to as the support set of
$\gSCOVdvs$.
The DOA estimates are the grid points, i.e., the elements of $\GRD$, 
which correspond to the $L$
largest peaks of $\gSCOVdvs$. 

\section{Extension to Correlated Sources}
\label{sec:extention}
In the previous section, 
we assumed that the sources impinging onto the system of subarrays
are uncorrelated. In this case,
the source covariance matrix satisfies (\ref{eq:Pk-uc})
and the measurement covariance matrix
reduces to (\ref{eq:Rk-uc}).
However, by dropping the assumption of uncorrelated sources, 
(\ref{eq:Pk-uc}) is no longer valid since the
matrix $\SCOVo$, defined in (\ref{eq:lambda-f}), is non-zero.
In this section, we extend the MLE, the SSR approach, and the CRB
which have been introduced in the previous section 
for the case of uncorrelated sources to the case of correlated sources.

\subsection{The MLE and SSR approaches for Correlated Sources}

The derivation of the MLE in the correlated sources case
is similar to the case of uncorrelated sources,
which is introduced in Section~\ref{sec:mle}. 
However, in this case, 
the off-diagonal entries of the source covariance matrix $\SCOV$
are non-zero.
Consequently,
the property (\ref{eq:Pk-uc}) does not hold.
Thus, in contrast to (\ref{eq:Rk-uc}),  
the measurement covariance matrix $\MCOV_k$
for correlated sources, defined in (\ref{eq:Rk}),
depends on the unknown displacements between the subarrays, 
represented by the matrix $\PHASE_k$
for $k=1,\ldots,K$.
The negative log-likelihood in the presence of correlated sources,
denoted as $\LogLH{ \DOAs , \SCOV, \NCOV, \PHASE_2, \ldots,\PHASE_K}$,
is defined in (\ref{eq:mle-uc}).
However, for $\LogLH{ \DOAs , \SCOV, \NCOV, \PHASE_2, \ldots,\PHASE_K}$
the covariance matrix as defined in (\ref{eq:Rk}) is used since (\ref{eq:Rk-uc})
is only valid for uncorrelated source.
The DOAs can be estimated from the minimization problem
\begin{equation}
\begin{aligned}
&\underset{\DOAs, \SCOV, \NCOV, \PHASE_1,\ldots,\PHASE_K}{\min}  
  \LogLH{ \DOAs , \SCOV, \NCOV, \PHASE_1,\ldots,\PHASE_K} \\
&{\rm s.t.} \quad  \SCOV \succeq 0, \\
& \quad\quad  \NCOV > 0,
\end{aligned}
 \label{eq:doa-C}
\end{equation} 
where $\SCOV \succeq 0$ 
denotes that the matrix $\SCOV$ is positive semidefinite.
Similar to the case of uncorrelated sources,
the optimization problem (\ref{eq:doa-C}) is nonconvex.
Therefore, starting from a feasible point,
a local solution (local minimum) for (\ref{eq:doa-C}) can be computed,
e.g., using the gradient descent method \cite{boyd2004convex}.

We remark that the SSR approach 
introduced in Section~\ref{sec:ssr-uc} for uncorrelated sources
is robust to the assumption of uncorrelated sources.
This robustness results from the fact that the SPICE method,
which we base our SSR approach on, is robust to the assumption of uncorrelated sources
\cite[Section~II]{stoica2011spice}.
Consequently, the SSR approach as introduced in  Section~\ref{sec:ssr-uc}
for uncorrelated sources is applicable in the case of correlated sources.

\subsection{The CRB for Correlated Sources}
\label{sec:corr-crb}
The derivation of the CRB for the case of correlated sources is similar to 
the case of uncorrelated sources. 
The CRB for the case of correlated sources is written as in (\ref{eq:crb-uc})
with  $\CrbPartTwo$ defined as
\begin{equation}
\CrbPartTwo \eq [ \D{\MCOVv} {\SCOVv^T}, \D{\MCOVv} {\NCOV},
\D{\MCOVv} {\DISP_2^T} 
\mdots
\D{\MCOVv} {\DISP_K^T}
],
\end{equation}
where $\SCOVv$ is a real vector of length $L^2$ which represents 
the unknown parameters
of the source covariance matrix.
More precisely $\SCOVv$ contains the diagonal of $\SCOV$ 
and the real and imaginary parts of the upper diagonal of the matrix $\SCOV$. 
In the following, %
we demonstrate the behaviour of the CRB at high SNR by simulation
and we carry out an asymptotic (for high SNR) analysis  of this behaviour.
\subsubsection*{Example}

In Fig.~\ref{fig:crb-corr}, we display the CRB for $K = 12$ subarrays each consists of two sensors
and $L = 2$ equally-powered correlated sources,\footnote{%
The same configuration as in the case {\caseCone} in Section~\ref{sec:id}, expect for the source correlation, is used.   
For the details on the array geometry parameters 
please refer the array setup described in Section~\ref{sec:simulation}. } i.e., 
the matrix $\Fim{}_k$, for $k=1,\ldots,K$, are not invertible.
Thus, the source covariance matrix is 
\begin{equation}
\SCOV = \SCOVds \; \MCORR
\end{equation}
where $\MCORR \eq\left[ \begin{matrix} 1 & \corr \\ \corr^* & 1 \end{matrix} \right]$, 
the correlation factor $\corr$ satisfies  $0 \leq |\corr| \leq 1$,
and $\SCOVds$ is the power of each of the two sources.
In Fig.~\ref{fig:crb-corr}, the CRB is displayed 
for correlation factor $\corr$ of  $0, 0.3, 0.6,$ and $1$, 
where the latter correlation value indicates coherent sources.
Observe in Fig.~\ref{fig:crb-corr} that the CRB of the estimated DOAs 
for correlated sources behaves similar to the uncorrelated sources case of Fig.~\ref{fig:crb}. 
However, the CRB decreases with the increase of $\corr$. 
Interestingly, for coherent sources, i.e., for $\corr=1$,
the CRB approaches zero at high SNR, which is in exact contrast to the case of uncorrelated or
 partly correlated sources where the CRB does not vanish with SNR.
\begin{figure}
\begin{tikzpicture}

\begin{axis}[%
width=\pFigW \columnwidth,
height=\pFigH  \columnwidth,
at={(0.979333in,0.628222in)},
scale only axis,
separate axis lines,
every outer x axis line/.append style={black},
every x tick label/.append style={font=\color{black}},
xmin=-8.7396,
xmax=35.9896,
xmajorgrids,
every outer y axis line/.append style={black},
every y tick label/.append style={font=\color{black}},
ymode=log,
ymin=0.0044,
ymax=2.2695,
yminorticks=false,
ymajorgrids,
xlabel={\pSnr},
ylabel={\pRmse},
legend style={at={(0.01,0.01)},anchor=south west,legend cell align=left,align=left,fill=white}
]

\addplot [color=mycolor2,solid,line width=\pLineW,mark options={line width=\pLineW}]
  table[row sep=crcr]{%
-30.0000	263.0516\\
-26.0000	105.0385\\
-22.0000	42.1318\\
-18.0000	17.0872\\
-14.0000	7.1142\\
-10.0000	3.1384\\
-6.0000	1.5458\\
-2.0000	0.8981\\
2.0000	0.6263\\
6.0000	0.5067\\
10.0000	0.4512\\
14.0000	0.4245\\
18.0000	0.4116\\
22.0000	0.4055\\
26.0000	0.4028\\
30.0000	0.4017\\
34.0000	0.4012\\
38.0000	0.4010\\
};
\addlegendentry{{\scriptsize CRB $\corr=0$}};

\addplot [color=mycolor1,solid,line width=\pLineW,mark=o,mark options={solid,draw=mycolor1,line width=\pLineW}]
  table[row sep=crcr]{%
-30.0000	323.7069\\
-28.0000	204.5035\\
-26.0000	129.2907\\
-24.0000	81.8340\\
-22.0000	51.8896\\
-20.0000	32.9941\\
-18.0000	21.0691\\
-16.0000	13.5407\\
-14.0000	8.7844\\
-12.0000	5.7746\\
-10.0000	3.8638\\
-8.0000	2.6436\\
-6.0000	1.8576\\
-4.0000	1.3455\\
-2.0000	1.0083\\
0.0000	0.7848\\
2.0000	0.6362\\
4.0000	0.5374\\
6.0000	0.4713\\
8.0000	0.4267\\
10.0000	0.3962\\
12.0000	0.3750\\
14.0000	0.3603\\
16.0000	0.3501\\
18.0000	0.3432\\
20.0000	0.3385\\
22.0000	0.3355\\
24.0000	0.3334\\
26.0000	0.3321\\
28.0000	0.3313\\
30.0000	0.3308\\
32.0000	0.3304\\
34.0000	0.3302\\
36.0000	0.3301\\
38.0000	0.3300\\
40.0000	0.3299\\
};
\addlegendentry{{\scriptsize CRB $\corr=0.3$}};

\addplot [color=blue,solid,line width=\pLineW,mark=x,mark options={solid,draw=blue,line width=\pLineW}]
  table[row sep=crcr]{%
-30.0000	283.0889\\
-28.0000	178.9012\\
-26.0000	113.1626\\
-24.0000	71.6831\\
-22.0000	45.5096\\
-20.0000	28.9924\\
-18.0000	18.5664\\
-16.0000	11.9816\\
-14.0000	7.8171\\
-12.0000	5.1756\\
-10.0000	3.4902\\
-8.0000	2.4030\\
-6.0000	1.6902\\
-4.0000	1.2136\\
-2.0000	0.8897\\
0.0000	0.6683\\
2.0000	0.5179\\
4.0000	0.4168\\
6.0000	0.3494\\
8.0000	0.3046\\
10.0000	0.2745\\
12.0000	0.2539\\
14.0000	0.2396\\
16.0000	0.2297\\
18.0000	0.2228\\
20.0000	0.2182\\
22.0000	0.2151\\
24.0000	0.2130\\
26.0000	0.2117\\
28.0000	0.2108\\
30.0000	0.2102\\
32.0000	0.2099\\
34.0000	0.2096\\
36.0000	0.2095\\
38.0000	0.2094\\
40.0000	0.2093\\
};
\addlegendentry{{\scriptsize CRB $\corr=0.6$}};

\addplot [color=red,solid,line width=\pLineW,mark=triangle,mark options={solid,rotate=180,draw=red,line width=\pLineW}]
  table[row sep=crcr]{%
-30.0000	216.5919\\
-26.0000	86.6521\\
-22.0000	34.9175\\
-18.0000	14.3104\\
-14.0000	6.0808\\
-10.0000	2.7522\\
-6.0000	1.3425\\
-2.0000	0.6854\\
2.0000	0.3543\\
6.0000	0.1890\\
10.0000	0.1066\\
14.0000	0.0629\\
18.0000	0.0382\\
22.0000	0.0236\\
26.0000	0.0147\\
30.0000	0.0092\\
34.0000	0.0058\\
38.0000	0.0037\\
};
\addlegendentry{{\scriptsize CRB $\corr=1$}};

\end{axis}
\end{tikzpicture}%
\caption{The CRB in the case of correlated sources for different source correlation $\corr$.}
\label{fig:crb-corr} 
\end{figure}
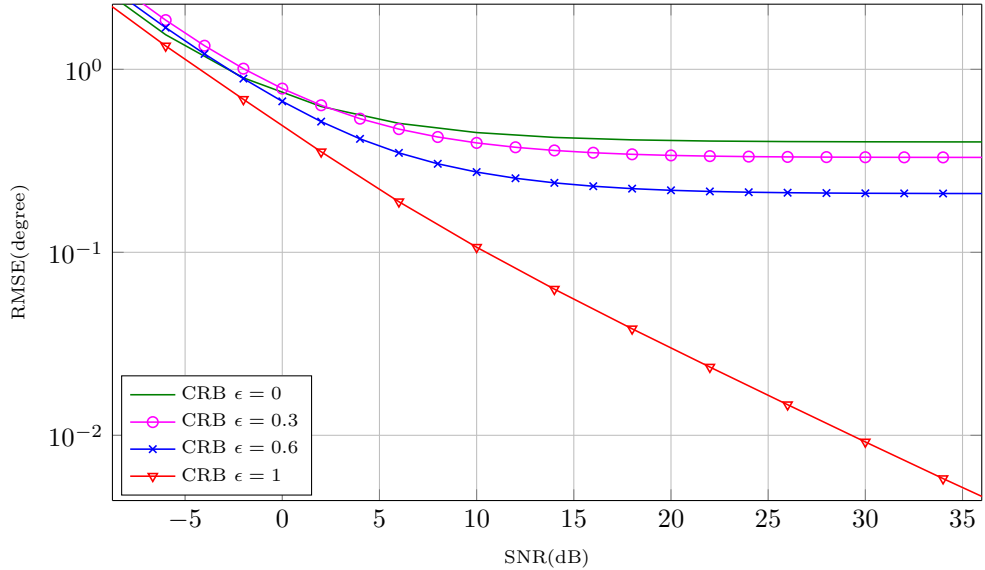 

In the sequel, the aforementioned behaviour of the CRB
is analyzed asymptotically for high SNR values.
Following the steps of Section~\ref{sec:crb}
for the case of uncorrelated sources, 
in the correlated source case the CRB at high SNR %
is written as
\begin{equation}
\Crbi %
 \approx N [\D{\CrbPartOneic}{\DOAs^T}]^H 
 \;\MCOVic\; 
[\D{\CrbPartOneic}{\DOAs^T}],
\label{eq:crb-infty-C} 
\end{equation}
where $\CrbPartOneic \eq [\CrbPartOneic_1^T,\ldots,\CrbPartOneic_K^T]^T$,
$\CrbPartOneic_k \eq \vec{}(\STRkn_k \PHASE_k \MCORR \PHASE_k^H  \STRkn_k^H)$,
$\MCOVic \eq \blkdiag \left( \MCOVic_1,\ldots,\MCOVic_K \right)$,
and
$\MCOVic_k \eq (\STRkn_k \PHASE_k \MCORR \PHASE_k^H \STRkn_k^H)^{-T}\otimes 
(\STRkn_k \PHASE_k \MCORR \PHASE_k^H \STRkn_k^H)^{-1}$.
Note that $\Crbi$ depends neither on $\SCOVds$
nor on $\NCOV$.
Thus, based on the rank of the matrix $\MCOVic_k^{-1}$, denoted as $\RNKi_k$,
the following two cases are distinguished:
\begin{itemize}
  \item [1)] The case when $|\corr| < 1$ in which $\RNKi_k=M_k^2$, for $k=1 \mdots K$, 
      consequently, the matrix $\MCOVic^{-1}$ is full rank, and the CRB does not vanish at high SNR.
  \item [2)] The case when $|\corr| = 1$ in which $\RNKi_1=\cdots\RNKi_K=1$,
  consequently, the matrix $\MCOVic^{-1}$ is rank deficient and asymptotically 
  drives the CRB to zero at high SNR.
\end{itemize}

\section{Simulation Results}
\label{sec:simulation}

\begin{figure*}
\subfloat[]{ 
\definecolor{mycolor1}{rgb}{1.00000,0.00000,1.00000}%
\begin{tikzpicture}

\begin{axis}[%
width=\pFigW \columnwidth,
height=\pFigH  \columnwidth,
scale only axis,
separate axis lines,
every outer x axis line/.append style={black},
every x tick label/.append style={font=\color{black}},
xmin=-30.0000,
xmax=40.0000,
xmajorgrids,
every outer y axis line/.append style={black},
every y tick label/.append style={font=\color{black}},
ymode=log,
ymin=0.1000,
ymax=1000.0000,
yminorticks=true,
ylabel={\pRmse},
xlabel={\pSnr},
ymajorgrids,
yminorgrids,
legend style={legend cell align=left,align=left,fill=white}
]

\addplot [color=blue,loosely dashed,line width=1.0pt
,mark=o,mark options={line width=0.6pt, solid}
, mark repeat={2}]
  table[row sep=crcr]{%
-30.0000	40.4993\\
-25.0000	38.2837\\
-20.0000	36.3803\\
-15.0000	26.2840\\
-10.0000	7.1679\\
-8.0000	5.2959\\
-6.0000	1.8170\\
-4.0000	1.1989\\
-2.0000	1.0983\\
0.0000	0.7878\\
2.0000	0.6832\\
4.0000	0.6566\\
6.0000	0.6446\\
8.0000	0.5965\\
10.0000	0.5978\\
12.0000	0.5776\\
14.0000	0.5712\\
16.0000	0.5406\\
18.0000	0.5698\\
20.0000	0.4927\\
22.0000	0.5253\\
24.0000	0.5064\\
26.0000	0.4330\\
28.0000	0.4704\\
30.0000	0.4698\\
32.0000	0.4698\\
34.0000	0.4518\\
36.0000	0.4752\\
38.0000	0.4895\\
40.0000	0.4615\\
};
\addlegendentry{\pSPICE};

\addplot [color=pMaxLikelihoodColor
, dotted
,mark=+,mark options={line width=1.0, solid}
, mark size= \pMarkSize pt
,mark repeat={2}
, mark phase=2
]
  table[row sep=crcr]{%
-30.0000	40.5002\\
-25.0000	38.4566\\
-20.0000	36.3791\\
-15.0000	26.1423\\
-10.0000	7.0641\\
-8.0000	5.2959\\
-6.0000	1.8170\\
-4.0000	1.1989\\
-2.0000	1.0983\\
0.0000	0.7878\\
2.0000	0.5404\\
4.0000	0.4865\\
6.0000	0.4854\\
8.0000	0.4394\\
10.0000	0.4150\\
12.0000	0.4471\\
14.0000	0.4799\\
16.0000	0.4668\\
18.0000	0.4952\\
20.0000	0.4118\\
22.0000	0.4606\\
24.0000	0.4680\\
26.0000	0.3537\\
28.0000	0.4129\\
30.0000	0.4155\\
32.0000	0.4274\\
34.0000	0.3989\\
36.0000	0.4154\\
38.0000	0.4258\\
40.0000	0.4218\\
};
\addlegendentry{\pMaxLikelihood};

\addplot [color=pCrbColor1,solid,line width=\pLineWS pt
, smooth
]
  table[row sep=crcr]{%
-30.0000	263.0516\\
-26.0000	105.0385\\
-22.0000	42.1318\\
-18.0000	17.0872\\
-14.0000	7.1142\\ 
-10.0000	3.1384\\
-6.0000		1.5458\\
-2.0000		0.8981\\
2.0000		0.6263\\
6.0000		0.5067\\
10.0000		0.4512\\
14.0000		0.4245\\
18.0000		0.4116\\
22.0000		0.4055\\
26.0000		0.4028\\
30.0000		0.4017\\
34.0000		0.4012\\
38.0000		0.4010\\
};
\addlegendentry{\pCrbUC};

\end{axis}
\end{tikzpicture}%
\label{fig:snr-uc} }
\hfill
\subfloat[]{ 
\definecolor{mycolor1}{rgb}{1.00000,0.00000,1.00000}%

\begin{tikzpicture}

\begin{axis}[%
width=\pFigW \columnwidth,
height=\pFigH  \columnwidth,
scale only axis,
separate axis lines,
every outer x axis line/.append style={black},
every x tick label/.append style={font=\pTickSize\color{black}},
xmin=-30.000,
xmax=40.000,
xmajorgrids,
xlabel={\pSnr},
every outer y axis line/.append style={black},
every y tick label/.append style={font=\pTickSize\color{black}},
ymin=0.000,
ymax=100.000,
ymajorgrids,
ylabel={\pResolution},
ylabel style={yshift=-10pt},
legend style={at={(0.98,0.38)}, legend cell align=left,align=left,fill=white}
]

\addplot [color=blue,loosely dashed,line width=1.0pt
,mark=o,mark options={line width=0.6pt, solid}
, mark repeat={2}]
  table[row sep=crcr]{%
-30.0000	2.0000\\
-25.0000	3.0000\\
-20.0000	10.0000\\
-15.0000	46.0000\\
-10.0000	96.0000\\
-8.0000	98.0000\\
-6.0000	100.0000\\
-4.0000	100.0000\\
-2.0000	100.0000\\
0.0000	100.0000\\
2.0000	100.0000\\
4.0000	100.0000\\
6.0000	100.0000\\
8.0000	100.0000\\
10.0000	100.0000\\
12.0000	100.0000\\
14.0000	100.0000\\
16.0000	100.0000\\
18.0000	100.0000\\
20.0000	100.0000\\
22.0000	100.0000\\
24.0000	100.0000\\
26.0000	100.0000\\
28.0000	100.0000\\
30.0000	100.0000\\
32.0000	100.0000\\
34.0000	100.0000\\
36.0000	100.0000\\
38.0000	100.0000\\
40.0000	100.0000\\
};
\addlegendentry{\pSPICE};

\addplot [color=pMaxLikelihoodColor
, dotted
,mark=+,mark options={line width=1.0, solid}
, mark size= \pMarkSize pt
,mark repeat={2}
, mark phase=2
]
  table[row sep=crcr]{%
-30.0000	2.0000\\
-25.0000	4.0000\\
-20.0000	10.0000\\
-15.0000	46.0000\\
-10.0000	96.0000\\
-8.0000	98.0000\\
-6.0000	100.0000\\
-4.0000	100.0000\\
-2.0000	100.0000\\
0.0000	100.0000\\
2.0000	100.0000\\
4.0000	100.0000\\
6.0000	100.0000\\
8.0000	100.0000\\
10.0000	100.0000\\
12.0000	100.0000\\
14.0000	100.0000\\
16.0000	100.0000\\
18.0000	100.0000\\
20.0000	100.0000\\
22.0000	100.0000\\
24.0000	100.0000\\
26.0000	100.0000\\
28.0000	100.0000\\
30.0000	100.0000\\
32.0000	100.0000\\
34.0000	100.0000\\
36.0000	100.0000\\
38.0000	100.0000\\
40.0000	100.0000\\
};
\addlegendentry{\pMaxLikelihood};

\end{axis}
\end{tikzpicture}%
 \label{fig:snr-pd-uc} }
 \caption{\small DOA estimation performance, assuming uncorrelated sources, plotted against {\snr}   
 for a fixed number of samples $N=50$: 
 \protect\subref{fig:snr-uc} the {\rmse} of the proposed DOA estimation methods
 averaged over $100$ realizations,
 \protect\subref{fig:snr-pd-uc}
 the resolution percentage of the proposed DOA estimation methods averaged over $100$ realizations.}
 \end{figure*}
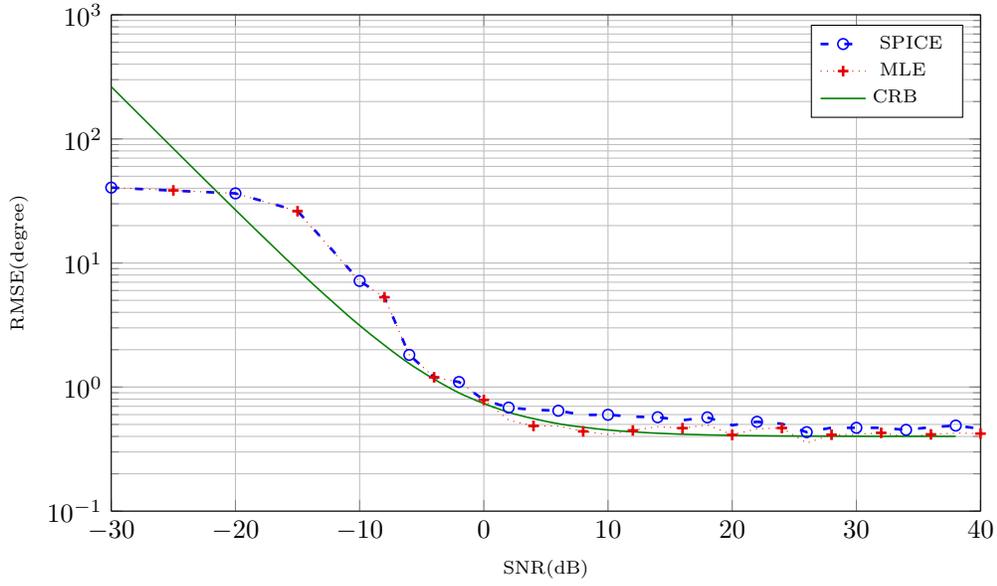
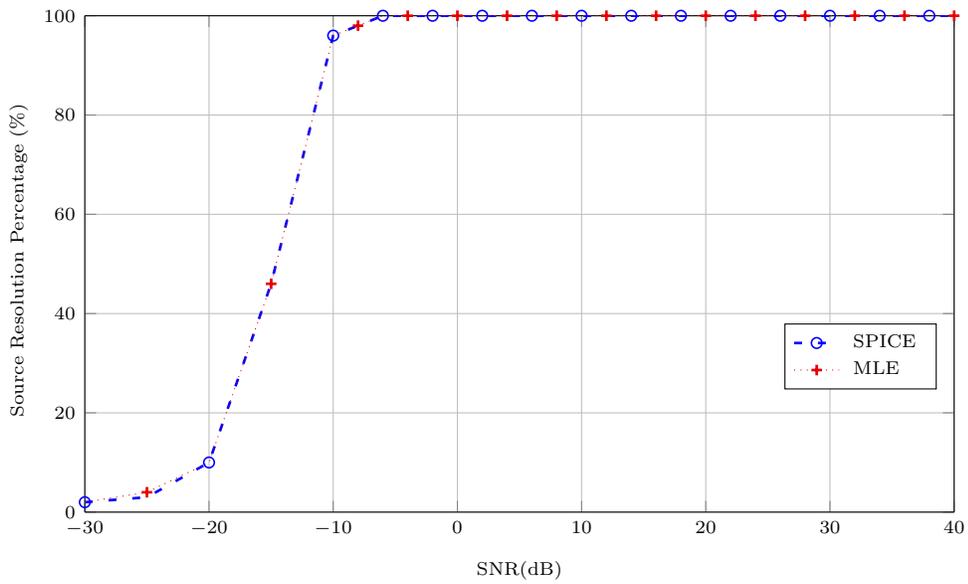
\begin{figure*}
\subfloat[]{ %
\definecolor{mycolor1}{rgb}{1.00000,0.00000,1.00000}%
\begin{tikzpicture}

\begin{axis}[%
width=\pFigW \columnwidth,
height=\pFigH  \columnwidth,
scale only axis,
separate axis lines,
every outer x axis line/.append style={black},
every x tick label/.append style={font=\pTickSize\color{black}},
scaled x ticks = false,
xmin=-2.748,
xmax=800.386,
xlabel={\pSamples},
xmajorgrids,
every outer y axis line/.append style={black},
every y tick label/.append style={font=\pTickSize\color{black}},
ymode=log,
ymin=0.010,
ymax=100.218,
yminorticks=true,
ylabel={\pRmse},
ymajorgrids,
yminorgrids,
ylabel style={yshift=-10pt},
legend style={at={(0.6,0.60)},anchor=south west,legend cell align=left,align=left,fill=white}
]
\addplot [color=blue,loosely dashed,line width=\pLineW
,mark repeat={2}
,mark phase={1}
,mark=o,mark options={line width=\pLineW, solid}]
  table[row sep=crcr]{%
5.0000		9.5044\\
10.0000		1.8990\\
20.0000		1.5060\\
60.0000		0.5765\\
70.0000		0.6100\\
80.0000		0.5558\\
90.0000		0.5375\\
100.0000	0.5277\\
200.0000	0.5008\\
300.0000	0.4555\\
400.0000	0.3922\\
500.0000	0.4297\\
600.0000	0.4047\\
700.0000	0.4000\\
800.0000	0.3991\\
900.0000	0.3834\\
1000.0000	0.3972\\
};
\addlegendentry{\pSPICE};

\addplot [color=pMaxLikelihoodColor
, dashed
, line width= \pLineW
,mark=+,mark options={line width=\pLineW, solid}
, mark size= \pMarkSize
,mark repeat={2}
,mark phase={2}
]
  table[row sep=crcr]{%
5.0000	9.5105\\
10.0000	1.7437\\
20.0000	1.4321\\
40.0000		0.9220\\
50.0000		0.8564\\
60.0000		0.7880\\
80.0000		0.5728\\
90.0000		0.5251\\
100.0000	0.4605\\
200.0000	0.4025\\
300.0000	0.3301\\
400.0000	0.2725\\
500.0000	0.3124\\
600.0000	0.2858\\
700.0000	0.2712\\
800.0000	0.2518\\
900.0000	0.2422\\
1000.0000	0.2006\\
};
\addlegendentry{\pMaxLikelihood};

\addplot [color=pCrbColor1,solid,line width=\pLineWS
, smooth
]
  table[row sep=crcr]{%
5.0000	2.8402\\ 
10.0000	2.0083\\
20.0000	1.4201\\
30.0000	1.1595\\
40.0000	1.0042\\
50.0000	0.8981\\
60.0000	0.8199\\
70.0000	0.7591\\
80.0000	0.7100\\
90.0000	0.6694\\
100.0000	0.6351\\
200.0000	0.4491\\
300.0000	0.3667\\
400.0000	0.3175\\
500.0000	0.2840\\
600.0000	0.2593\\
700.0000	0.2400\\
800.0000	0.2245\\
900.0000	0.2117\\
1000.0000	0.2008\\
};
\addlegendentry{\pCrbUC};

\end{axis}
\end{tikzpicture}
\label{fig:n-uc}}
\hfill
\subfloat[]{ 
\begin{tikzpicture}

\begin{axis}[%
width=\pFigW \columnwidth,
height=\pFigH  \columnwidth,
scale only axis,
separate axis lines,
every outer x axis line/.append style={black},
every x tick label/.append style={font=\color{black}},
xmin=0.000,
xmax=200.000,
xmajorgrids,
xlabel={\pSamples},
every outer y axis line/.append style={black},
every y tick label/.append style={font=\color{black}},
ymin=60.000,
ymax=100.000,
ymajorgrids,
ylabel={\pResolution},
ylabel style={yshift=-10pt},
legend style={at={(0.98,0.28)}, legend cell align=left,align=left,fill=white}
]
\addplot [color=blue,loosely dashed,line width=1.0pt
,mark=o,mark options={line width=0.6pt, solid}
, mark repeat={2}]
  table[row sep=crcr]{%
5.0000	94.0000\\
10.0000	100.0000\\
20.0000	100.0000\\
30.0000	100.0000\\
40.0000	100.0000\\
50.0000	100.0000\\
60.0000	100.0000\\
70.0000	100.0000\\
80.0000	100.0000\\
90.0000	100.0000\\
100.0000	100.0000\\
200.0000	100.0000\\
300.0000	100.0000\\
400.0000	100.0000\\
500.0000	100.0000\\
600.0000	100.0000\\
700.0000	100.0000\\
800.0000	100.0000\\
900.0000	100.0000\\
1000.0000	100.0000\\
};
\addlegendentry{\pSPICE};

\addplot [color=pMaxLikelihoodColor
, dotted
,mark=+,mark options={line width=1.0, solid}
, mark size= \pMarkSize pt
,mark repeat={2}
, mark phase=2
]
  table[row sep=crcr]{%
5.0000	92.0000\\
10.0000	100.0000\\
20.0000	100.0000\\
30.0000	100.0000\\
40.0000	100.0000\\
50.0000	100.0000\\
60.0000	100.0000\\
70.0000	100.0000\\
80.0000	100.0000\\
90.0000	100.0000\\
100.0000	100.0000\\
200.0000	100.0000\\
300.0000	100.0000\\
400.0000	100.0000\\
500.0000	100.0000\\
600.0000	100.0000\\
700.0000	100.0000\\
800.0000	100.0000\\
900.0000	100.0000\\
1000.0000	100.0000\\
};
\addlegendentry{\pMaxLikelihood};

\end{axis}
\end{tikzpicture} \label{fig:n-pd-uc} }
\caption{ %
\small DOA estimation performance, assuming uncorrelated sources, plotted against the number of snapshots $N$   
 for a fixed {\snr}$=-2$ dB: 
 \protect\subref{fig:n-uc} the {\rmse} of the proposed DOA estimation methods
 averaged over $100$ realizations,
 \protect\subref{fig:n-pd-uc}
 the resolution percentage of the proposed DOA estimation methods averaged over $100$ realizations.}
\end{figure*}
\begin{figure*}
\subfloat[]{%

\definecolor{mycolor1}{rgb}{1.00000,0.00000,1.00000}%
\begin{tikzpicture}

\begin{axis}[%
width=\pFigW \columnwidth,
height=\pFigH  \columnwidth,
x dir=reverse,
scale only axis,
separate axis lines,
every outer x axis line/.append style={black},
every x tick label/.append style={font=\pTickSize\color{black}},
xmin=1.0000,
xmax=8.0000,
xlabel={\pL},
xmajorgrids,
every outer y axis line/.append style={black},
every y tick label/.append style={font=\pTickSize\color{black}},
ymode=log,
ymin=0.100,
ymax=10000.0000,
yminorticks=true,
ylabel={\pRmse},
ylabel style={yshift=-10pt},
ymajorgrids,
yminorgrids,
legend style={at={(0.5
5,0.68)},anchor=south west,legend cell align=left,align=left,fill=white}
]
\addplot [color=blue,loosely dashed,line width=1.0pt
,mark=o,mark options={line width=0.6pt, solid}
]
  table[row sep=crcr]{%
1.0000	0.2464\\
2.0000	0.3798\\
3.0000	0.8243\\
4.0000	1.7986\\
5.0000	3.5299\\
6.0000	11.4878\\
7.0000	18.3532\\
8.0000	24.4103\\
};
\addlegendentry{\pSPICE};

\addplot [color=pMaxLikelihoodColor
, dotted
,mark=+,mark options={line width=1.0, solid}
, mark size= \pMarkSize pt
]
  table[row sep=crcr]{%
1.0000	0.1905\\
2.0000	0.3279\\
3.0000	0.5360\\
4.0000	0.8922\\
5.0000	2.5672\\
6.0000	11.2414\\
7.0000	18.3290\\
8.0000	24.0646\\
};
\addlegendentry{\pMaxLikelihood};

\addplot [color=pCrbColor1,solid,line width=\pLineWS pt
]
  table[row sep=crcr]{%
1.0000	0.1916\\
2.0000	0.2609\\
3.0000	0.4131\\
4.0000	0.5971\\
5.0000	1.0837\\
6.0000	1.5321\\
7.0000	55.2461\\
8.0000	2595.2865\\
};
\addlegendentry{\pCrbUC};

\end{axis}
\end{tikzpicture} \label{fig:l-uc}}
\hfill
\subfloat[]{%
 
\begin{tikzpicture}

\begin{axis}[%
width=\pFigW \columnwidth,
height=\pFigH  \columnwidth,
at={(0.808889in,0.513333in)},
x dir=reverse,
scale only axis,
separate axis lines,
every outer x axis line/.append style={black},
every x tick label/.append style={font=\pTickSize\color{black}},
xmin=1.0000,
xmax=8.0000,
xlabel={\pL},
xmajorgrids,
every outer y axis line/.append style={black},
every y tick label/.append style={font=\pTickSize\color{black}},
ymin=0.0000,
ymax=100.0000,
ylabel={\pResolution},
ylabel style={yshift=-10pt},
ymajorgrids,
legend style={at={(0.45,0.05)}, anchor=south west, legend cell align=right,fill=white}
]
\addplot [color=blue,loosely dashed,line width=1.0pt
,mark=o,mark options={line width=0.6pt, solid}]
  table[row sep=crcr]{%
1.0000	100.0000\\
2.0000	100.0000\\
3.0000	100.0000\\
4.0000	100.0000\\
5.0000	100.0000\\
6.0000	0.0000\\
7.0000	0.0000\\
8.0000	0.0000\\
};
\addlegendentry{\pSPICE};

\addplot [color=pMaxLikelihoodColor
, dotted
,mark=+,mark options={line width=1.0, solid}
, mark size= \pMarkSize pt
,mark repeat={1}
]
  table[row sep=crcr]{%
1.0000	100.0000\\
2.0000	100.0000\\
3.0000	100.0000\\
4.0000	100.0000\\
5.0000	100.0000\\
6.0000	0.0000\\
7.0000	0.0000\\
8.0000	0.0000\\
};
\addlegendentry{\pMaxLikelihood};

\end{axis}
\end{tikzpicture} \label{fig:l-pd-uc}}
 \caption{\small DOA estimation performance, assuming uncorrelated sources, plotted against the number of sources $L$   
 for a fixed number of samples $N=50$ and a fixed {\snr}$=-2$ dB: 
 \protect\subref{fig:l-uc} the {\rmse} of the proposed DOA estimation methods
 averaged over $100$ realizations,
 \protect\subref{fig:l-pd-uc}
 the resolution percentage of the proposed DOA estimation methods averaged over $100$ realizations.} 
\end{figure*}

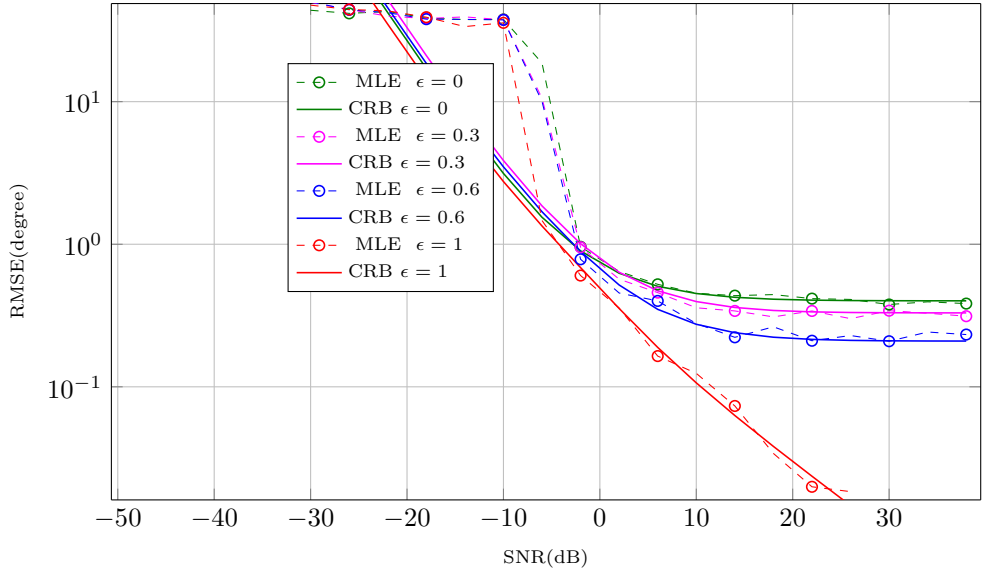
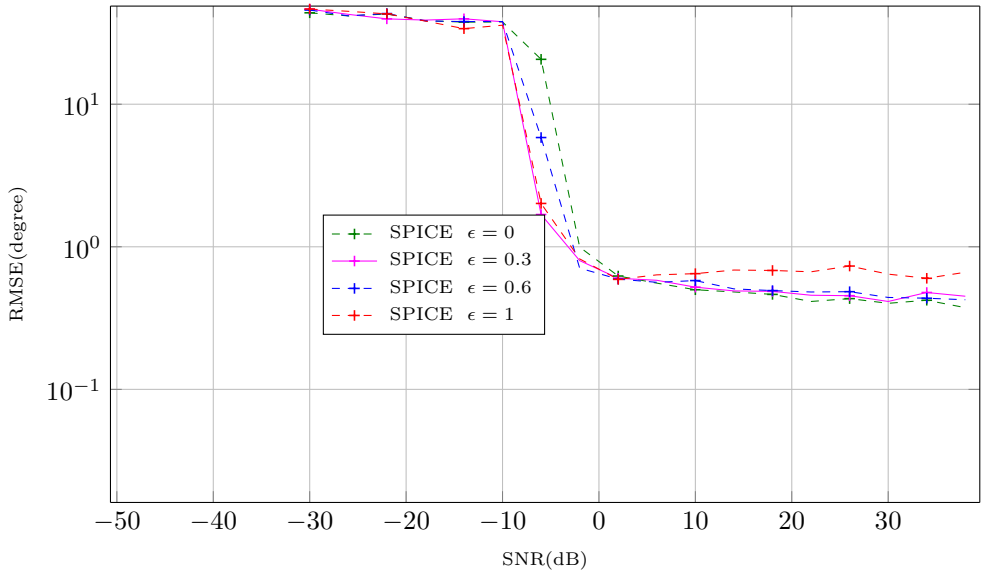
\begin{figure*}[]
\subfloat[]{
\definecolor{mycolor1}{rgb}{1.00000,0.00000,1.00000}%
\definecolor{mycolor2}{rgb}{0.00000,0.49804,0.00000}%
\begin{tikzpicture}

\begin{axis}[%
width=\pFigW \columnwidth,
height=\pFigH  \columnwidth,
at={(1.892222in,0.862889in)},
scale only axis,
separate axis lines,
every outer x axis line/.append style={black},
every x tick label/.append style={font=\color{black}},
xmin=-50.6644,
xmax=39.5000,
xmajorgrids,
every outer y axis line/.append style={black},
every y tick label/.append style={font=\color{black}},
ymode=log,
ymin=0.0161,
ymax=48.7776,
yminorticks=false,
ylabel={\pRmse},
xlabel={\pSnr},
ymajorgrids,
legend style={at={(0.44,0.88)}, legend cell align=left,align=left,draw=black}
]

\addplot [color=mycolor2,dashed
,mark=o,mark options={line width=\pLineW, solid}
, mark size= \pMarkSize 
,mark repeat={2},
mark phase={2}]
  table[row sep=crcr]{%
-30.0000	43.7693\\
-26.0000	41.6258\\
-22.0000	42.9289\\
-18.0000	38.5935\\
-14.0000	37.6511\\
-10.0000	37.7504\\
-6.0000		18.6568\\
-2.0000		0.9597\\
2.0000		0.6447\\
6.0000		0.5235\\
10.0000		0.4534\\
14.0000		0.4360\\
18.0000		0.4435\\
22.0000		0.4156\\
26.0000		0.4101\\
30.0000		0.3779\\
34.0000		0.3946\\
38.0000		0.3835\\
};
\addlegendentry{\pTextSize{{\pMaxLikelihood} $\corr = 0$}};

\addplot [color=mycolor2,solid,line width=\pLineW]
  table[row sep=crcr]{%
-30.0000	263.0516\\
-26.0000	105.0385\\
-22.0000	42.1318\\
-18.0000	17.0872\\
-14.0000	7.1142\\ 
-10.0000	3.1384\\
-6.0000		1.5458\\
-2.0000		0.8981\\
2.0000		0.6263\\
6.0000		0.5067\\
10.0000		0.4512\\
14.0000		0.4245\\
18.0000		0.4116\\
22.0000		0.4055\\
26.0000		0.4028\\
30.0000		0.4017\\
34.0000		0.4012\\
38.0000		0.4010\\
};
\addlegendentry{\pTextSize{{\pCrbUC} $\corr = 0$}};

\addplot [color=mycolor1, dashed
,mark=o,mark options={line width=\pLineW, solid}
, mark size= \pMarkSize 
,mark repeat={2},
mark phase={1}]
  table[row sep=crcr]{%
-30.0000	49.6851\\
-26.0000	44.2857\\
-22.0000	39.7942\\
-18.0000	38.3306\\
-14.0000	39.3170\\
-10.0000	37.5370\\
-6.0000		10.5809\\
-2.0000		0.9529\\
2.0000		0.5688\\
6.0000		0.4578\\
10.0000		0.3584\\
14.0000		0.3406\\
18.0000		0.3095\\
22.0000		0.3405\\
26.0000		0.3007\\
30.0000		0.3420\\
34.0000		0.3284\\
38.0000		0.3128\\
};
\addlegendentry{\pTextSize{{\pMaxLikelihood} $\corr = 0.3$}};
\addplot [color=mycolor1,solid,line width=\pLineW]
  table[row sep=crcr]{%
-30.0000	323.7069\\
-26.0000	129.2907\\
-22.0000	51.8896\\
-18.0000	21.0691\\
-14.0000	8.7844\\
-10.0000	3.8638\\
-6.0000		1.8576\\
-2.0000		1.0083\\
2.0000		0.6362\\
6.0000		0.4713\\
10.0000		0.3962\\
14.0000		0.3603\\
18.0000		0.3432\\
22.0000		0.3355\\
26.0000		0.3321\\
30.0000		0.3308\\
34.0000		0.3302\\
38.0000		0.3300\\
};
\addlegendentry{\pTextSize{{\pCrbUC} $\corr = 0.3$}};

\addplot [color=blue,dashed, 
,mark=o,mark options={line width=\pLineW, solid}
, mark size= \pMarkSize 
,mark repeat={2},
mark phase={1}]
  table[row sep=crcr]{%
-30.0000	50.1301\\
-26.0000	44.0936\\
-22.0000	42.0886\\
-18.0000	37.9004\\
-14.0000	37.7749\\
-10.0000	37.4818\\
-6.0000		10.0443\\
-2.0000		0.7848\\
2.0000		0.4549\\
6.0000		0.4004\\
10.0000		0.2779\\
14.0000		0.2224\\
18.0000		0.2636\\
22.0000		0.2105\\
26.0000		0.2289\\
30.0000		0.2088\\
34.0000		0.2423\\
38.0000		0.2326\\
};
\addlegendentry{\pTextSize{{\pMaxLikelihood} $\corr = 0.6$}};

\addplot [color=blue,solid,line width=\pLineW]
  table[row sep=crcr]{%
-30.0000	283.0889\\
-26.0000	113.1626\\
-22.0000	45.5096\\
-18.0000	18.5664\\
-14.0000	7.8171\\
-10.0000	3.4902\\
-6.0000		1.6902\\
-2.0000		0.8897\\
2.0000		0.5179\\
6.0000		0.3494\\
10.0000		0.2745\\
14.0000		0.2396\\
18.0000		0.2228\\
22.0000		0.2151\\
26.0000		0.2117\\
30.0000		0.2102\\
34.0000		0.2096\\
38.0000		0.2094\\
};
\addlegendentry{\pTextSize{{\pCrbUC} $\corr = 0.6$}};

\addplot [color=red,dashed, 
,mark=o,mark options={line width=\pLineW, solid}
, mark size= \pMarkSize 
,mark repeat={2},
mark phase={2}]
  table[row sep=crcr]{%
-30.0000	47.5255\\
-26.0000	44.2968\\
-22.0000	43.3540\\
-18.0000	39.1470\\
-14.0000	33.6284\\
-10.0000	35.7236\\
-6.0000		1.4671\\
-2.0000		0.6031\\
2.0000		0.3572\\
6.0000		0.1643\\
10.0000		0.1243\\
14.0000		0.0734\\
18.0000		0.0341\\
22.0000		0.0199\\
26.0000		0.0184\\
};
\addlegendentry{\pTextSize{{\pMaxLikelihood} $\corr = 1$}};

\addplot [color=red,solid,line width=\pLineW]
  table[row sep=crcr]{%
-30.0000	216.5919\\
-26.0000	86.6521\\
-22.0000	34.9175\\
-18.0000	14.3104\\
-14.0000	6.0808\\
-10.0000	2.7522\\
-6.0000		1.3425\\
-2.0000		0.6854\\
2.0000		0.3543\\
6.0000		0.1890\\
10.0000		0.1066\\
14.0000		0.0629\\
18.0000		0.0382\\
22.0000		0.0236\\
26.0000		0.0147\\
30.0000		0.0092\\
34.0000		0.0058\\
38.0000		0.0037\\
};
\addlegendentry{\pTextSize{{\pCrbUC} $\corr = 1$}};

\end{axis}
\end{tikzpicture} \label{fig:snr-c}}
\hfill
\subfloat[]{%
\definecolor{mycolor1}{rgb}{1.00000,0.00000,1.00000}%
\definecolor{mycolor2}{rgb}{0.00000,0.49804,0.00000}%
\begin{tikzpicture}

\begin{axis}[%
width=\pFigW \columnwidth,
height=\pFigH  \columnwidth,
at={(1.892222in,0.862889in)},
scale only axis,
separate axis lines,
every outer x axis line/.append style={black},
every x tick label/.append style={font=\color{black}},
xmin=-50.6644,
xmax=39.5000,
xmajorgrids,
every outer y axis line/.append style={black},
every y tick label/.append style={font=\color{black}},
ymode=log,
ymin=0.0161,
ymax=48.7776,
yminorticks=false,
ylabel={\pRmse},
xlabel={\pSnr},
ymajorgrids,
legend style={at={(0.50,0.58)}, legend cell align=left,align=left,draw=black}
]

\addplot [color=mycolor2,dashed
,mark=+,mark options={line width=\pLineW, solid}
, mark size= \pMarkSize pt
,mark repeat={2}]
  table[row sep=crcr]{%
-30.0000	43.7693\\
-26.0000	41.6258\\
-22.0000	42.9287\\
-18.0000	38.5935\\
-14.0000	37.6512\\
-10.0000	37.7504\\
-6.0000		20.6483\\
-2.0000		0.9947\\
2.0000		0.6238\\
6.0000		0.5561\\
10.0000		0.4998\\
14.0000		0.4822\\
18.0000		0.4636\\
22.0000		0.4135\\
26.0000		0.4330\\
30.0000		0.4014\\
34.0000		0.4218\\
38.0000		0.3743\\
};
\addlegendentry{\pTextSize{{\pSPICE} $\corr = 0$}};

\addplot [color=mycolor1,solid
,mark=+,mark options={line width=\pLineW, solid}
, mark size= \pMarkSize
,mark repeat={2},
mark phase={1}]
  table[row sep=crcr]{%
-30.0000	46.4108\\
-26.0000	42.6970\\
-22.0000	39.6907\\
-18.0000	38.9162\\
-14.0000	39.7372\\
-10.0000	37.8802\\
-6.0000		1.6843\\
-2.0000		0.8000\\
2.0000		0.6012\\
6.0000		0.5822\\
10.0000		0.5211\\
14.0000		0.4904\\
18.0000		0.4868\\
22.0000		0.4572\\
26.0000		0.4539\\
30.0000		0.4135\\
34.0000		0.4780\\
38.0000		0.4500\\
};
\addlegendentry{\pTextSize{{\pSPICE} $\corr = 0.3$}};

\addplot [color=blue,dashed, 
,mark=+,mark options={line width=\pLineW, solid}
, mark size= \pMarkSize
,mark repeat={2},
mark phase={1}]
  table[row sep=crcr]{%
-30.0000	45.7706\\
-26.0000	41.6696\\
-22.0000	43.2315\\
-18.0000	38.3650\\
-14.0000	37.9941\\
-10.0000	38.0062\\
-6.0000		5.8383\\
-2.0000		0.7083\\
2.0000		0.5933\\
6.0000		0.5635\\
10.0000		0.5794\\
14.0000		0.5067\\
18.0000		0.4930\\
22.0000		0.4819\\
26.0000		0.4843\\
30.0000		0.4413\\
34.0000		0.4359\\
38.0000		0.4251\\
};
\addlegendentry{\pTextSize{{\pSPICE} $\corr = 0.6$}};

\addplot [color=red,dashed, 
,mark=+,mark options={line width=\pLineW, solid}
, mark size= \pMarkSize 
,mark repeat={2},
mark phase={1}]
  table[row sep=crcr]{%
-30.0000	46.5312\\
-26.0000	44.8334\\
-22.0000	43.0192\\
-18.0000	38.2872\\
-14.0000	33.8922\\
-10.0000	35.8312\\
-6.0000		2.0138\\
-2.0000		0.8222\\
2.0000		0.5908\\
6.0000		0.6356\\
10.0000		0.6481\\
14.0000		0.6870\\
18.0000		0.6823\\
22.0000		0.6678\\
26.0000		0.7328\\
30.0000		0.6430\\
34.0000		0.6012\\
38.0000		0.6629\\
};
\addlegendentry{\pTextSize{{\pSPICE} $\corr = 1$}};

\end{axis}
\end{tikzpicture} \label{fig:snr-ssr-c}}
\caption{\small {\rmse} for correlated sources as a function of $\snr$
for a fixed number of samples $N=50$
for different values of the correlation factor $\corr$:
\protect\subref{fig:snr-c} the averaged {\rmse} for the MLE approach,
\protect\subref{fig:snr-ssr-c} the averaged {\rmse} for the SPICE approach.
}
\end{figure*}

In this section, we analyze the performance of 
our proposed MLE and SSR estimation methods
using simulations for both cases uncorrelated and correlated sources.
The cvx \cite{cvx} framework is used to solve 
the SPICE optimization problem in (\ref{eq:spice}),
where the the field-of-view is sampled every $0.1^{\circ}$.

The MLE is initialized with the solution of the SPICE method and the MATLAB command
\emph{fmincon} is used to compute the MLE 
as presented in (\ref{eq:doa-uc}) and (\ref{eq:doa-C}) for uncorrelated and correlated sources, respectively.
 
In our simulations,
an array composed of $K=12$ subarrays each is comprised of $2$ sensors is considered.
The location of the first sensors in the $12$ subarrays
measured in half-wavelength are
$(0,         0)$,
$(17.3,    6)$,
$(-2.4,   6.2)$,
$(10.5,   -2)$,
$(12.7,    2.1)$,
$(4.6,   -2.4)$,
$(4.6,    4.5)$,
$(4.5,    5.3)$,
$(2.3,    9)$,
$(10.2,    8.1)$,
$(10.2,    4)$,
and
$(13.4,    6)$.
These locations are considered to be unknown 
during the DOA estimation process.
The locations of the second sensors in each subarray
with respect to the first sensor in the corresponding subarray
measured in half-wavelength are
$(6.5, 0)$,
$(4.4, 0)$,
$(3.5, 0)$,
$(2.6, 0)$,
$(2.6, 0)$,
$(2.5, 0)$,
$(1.9, 0)$,
$(1.5, 0)$,
$(1.4, 0)$,
$(1.3, 0)$,
$(1, 0)$,
and
$(0.5, 0)$. 
These locations are considered to be known.
Signals of two far-field equal-powered uncorrelated sources 
are impinging onto the subarrays
from directions $-11.4^{\circ}$ $-1.1^{\circ}$.
In our simulations, the root mean square error (RMSE)
for the estimated DOAs is computed over $100$ realizations as
\begin{equation}
 \big( \frac{1}{100} \sum_{i=1}^{100} \frac{1}{L} \sum_{l=1}^L (\hat{\DOA}_l(i) - \DOA_l)^2  \big)^{1/2},
\label{eq:rmse}
\end{equation}
where $\hat{\DOA}_l(i)$ is the estimate of the $l$th DOA at realization $i$.
The RMSE in (\ref{eq:rmse}) is computed for the SPICE and the MLE approaches.
We also display the CRB computed as
\begin{equation}
\big( \frac{1}{L} \sum_{i=1}^L [\Crb_{\DOAs}]_{l,l}  \big)^{1/2},
\end{equation}
where $[\Crb_{\DOAs}]_{l,l}$ is the $l$th diagonal entry of the matrix $\Crb_{\DOAs}$. 

In Fig.~\ref{fig:snr-uc},
the averaged performance of the SPICE and the MLE
for a fixed number of samples $N=50$
is plotted against SNR.
It can be observed in Fig.~\ref{fig:snr-uc}
that the MLE  and the SPICE method 
achieves the CRB at high SNR.
In Fig.~\ref{fig:snr-pd-uc},
the source resolution percentage of the considered DOA estimation methods is plotted against the SNR,
where two sources are considered to be resolved if the 
error in the estimated DOAs is less than half of the angular separation between the 
two sources \cite{Parvazi2011}.
Observe that for $\snr \geq -8$ dB, the MLE and SPICE method can always identify the sources
and for $\snr \leq -20$ dB the resolution percentage is almost zero.

In Fig.~\ref{fig:n-uc}, the RMSE of DOA estimation using SPICE and MLE is 
plotted against the number of snapshots $N$ for a fixed $\snr=-2$ dB.
The MLE achieves the CRB for $N \geq 20$ samples,
whereas the SPICE method is above the CRB because of the bias resulting from the
nature of the SSR approaches \cite{malioutov2005sparse}.
In Fig.~\ref{fig:n-pd-uc}, 
the source resolution percentage is plotted against $N$.
Observe that the SPICE and the MLE achieve $100\%$ resolution percentage 
for $N \geq 20$.

In Fig.~\ref{fig:l-uc} and Fig.~\ref{fig:l-pd-uc},
for a fixed $\snr$ of $-2$ dB and fixed number of samples $N=50$, 
the number of sources $L$ is changed.
The source DOAs are chosen in order from the set $\{15^{\circ}$, $-15^{\circ}$, 
$30^{\circ}$, $-30^{\circ}$, $45^{\circ}$, $-45^{\circ}$, 
$60^{\circ}$, $-60^{\circ}\}$. 
Observe in Fig.~\ref{fig:l-uc} that for small number of 
sources $L \leq 4$ the MLE and the SPICE achieves the CRB. 
In Fig.~\ref{fig:l-pd-uc},
it can be seen that for $L \leq 5$ both the SPICE and the MLE methods 
are always able to identify the sources.
We remark that since $M_k=2$ for $k=1 \mdots K$
none of the subarrays can individually identify more than one source,
however, with our proposed methods,
which exploit the diverse structure of the subarrays, 
up to $L=5$ sources can be identified.

In the following,
we investigate the performance of the MLE and SPICE 
considering $L=2$  correlated sources. 
In Fig.~\ref{fig:snr-c}, the number of samples is fixed to $N=50$ 
and the {\rmse} for DOA estimation of the MLE is plotted against $\snr$ 
for different values of the correlation factor $\corr=0, 0.3, 0.6$, and $\corr = 1$.
Note that the RMSE decreases by increasing $\corr$.
For coherent sources, i.e., $\corr=1$,
the RMSE approaches zero for high SNR,
which is in correspondence to our discussion in Section~\ref{sec:corr-crb}.
The averaged performance of the SPICE for the same scenario is shown in Fig.~\ref{fig:snr-ssr-c}.
Note that the SPICE method is robust against the assumption of correlated sources,
i.e., 
the performance of SPICE does not degrade much with the increased correlation between the sources, 
see \cite{stoica2011spice}.
\section{Summary and Conclusions}
In this paper, we considered non-coherent DOA estimation using partly calibrated arrays.
The presentation is focused on the case where none of the subarrays
is able to individually identify all the sources.
A sufficient condition for uncorrelated sources identifiability
using non-coherent processing is presented.
We proved that using non-coherent processing it is possible to
identify more sources than each subarray individually can.   
Moreover, the CRB for non-coherent processing is derived 
and its behaviour at high SNR is analyzed.
Two methods, namely the MLE and SPICE,
are proposed to estimate the DOAs
from the sample covariance matrices received from all subarrays. 
Using the simulations, the performance of the MLE is shown to  achieve the derived CRB.

\section{Acknowledgements}
The project ADEL acknowledges the financial support of the Seventh
Framework Programme for Research of the European Commission under
grant number: 619647.
\appendix
\section{Proof of Theorem~\ref{thm:id}}
\label{apx:thm-id}
The proof of Theorem~\ref{thm:id} consists in showing the 
sufficiency of the condition (\ref{eq:id-nc}).
We remark that for fully calibrated arrays using coherent processing
a bound on the maximum number of identifiable sources
is introduced in \cite{hochwald1996identifiability}.
This bound is not applicable in our case since in \cite{hochwald1996identifiability} 
the covariance matrix of the whole array is assumed to be available
and thus the bound is introduced using the rank of the matrix $\STRkn$ 
and not $\bSTRkn$ as is this paper.
Our proof of the bound is similar in spirit to that 
of \cite{hochwald1996identifiability}. 

\subsection{Sufficiency of (\ref{eq:id-nc})}
In this section,  
we prove that if
$\bSTRkn (\DOAs) \SCOVdv = \bSTRkn (\DOAs') \SCOVdv'$
and
$L \leq \Floor{ \frac{\RNK}{2}}$ then
$\DOAs =  \DOAs'$. 
\begin{proof}
Assume that there are $q \leq L \leq \Floor{\frac{\RNK}{2}}$ 
entries which occur in both DOA vectors 
$\DOAs$ and $\DOAs'$.
Then, $\DOAs$ and $\DOAs'$
can be split as
$\DOAs= [\DOAs_1^T, \DOAs_2^T]^T$
and
$\DOAs'= [\DOAs_1'^T, \DOAs_2'^T]^T$
such that $\DOAs_1 = \DOAs_1' \in \reals^{q \times 1}$ 
and that the DOAs $\DOAs_2$ and
$\DOAs_2'$ are all different.
Moreover, we define $\SCOVdv=[\SCOVdv_1^T, \SCOVdv_2^T]^T$ 
and 
$\SCOVdv'=[\SCOVdv_1'^T, \SCOVdv_2'^T]^T$,
where  
$\SCOVdv_1, \SCOVdv_2, \SCOVdv_1'$,
and
$\SCOVdv_2'$
contain the power of the sources corresponding to the
DOAs 
$\DOAs_1, \DOAs_2, \DOAs_1'$
and
$\DOAs_2'$, respectively.
Thus, the assumption that
$\bSTRkn (\DOAs) \SCOVdv = \bSTRkn (\DOAs') \SCOVdv'$
can be written as
\begin{equation}
[\bSTRkn (\DOAs_1), \bSTRkn (\DOAs_2) ]
[\SCOVdv_1^T, \SCOVdv_2^T]^T
= 
[\bSTRkn (\DOAs_1'), \bSTRkn ( \DOAs_2' )]
[\SCOVdv_1'^T, \SCOVdv_2'^T]^T.
\label{eq:id-proof-1}
\end{equation}
Since 
$\bSTRkn (\DOAs_1) = \bSTRkn (\DOAs_1')$,
(\ref{eq:id-proof-1}) can be rearranged as
\begin{equation}
[\bSTRkn (\DOAs_1), \bSTRkn (\DOAs_2), \bSTRkn (\DOAs_2')]
[\SCOVdv_1^T - \SCOVdv_1'^T, \SCOVdv_2^T, -\SCOVdv_2'^T]^T = 0.
\label{eq:id-proof-2}
\end{equation}
Next, we distinguish between the following two cases:
\begin{enumerate}
  \item $q=L$: In this case $\DOAs_1=\DOAs=\DOAs'$ and $\SCOVdv = \SCOVdv'$
		is a unique solution to (\ref{eq:id-proof-2}), i.e., 
		in this case the DOAs are uniquely identifiable.
  \item $q < L$: In this case,
the matrix $[\bSTRkn (\DOAs_1)$, $\bSTRkn (\DOAs_2)$, $\bSTRkn (\DOAs_2') ]$
contains $2L-q$ columns corresponding to different DOAs.
Since $q < L$ and $L \leq \Floor{\frac{\RNK{}}{2}}$ 
the inequality $2L-q \leq 2L \leq \RNK{}$ holds.
Consequently, the matrix
$[\bSTRkn (\DOAs_1), \bSTRkn (\DOAs_2), \bSTRkn (\DOAs_2') ]$
is full rank
and (\ref{eq:id-proof-2}) can only be satisfied,
in this case,
if $[\SCOVdv_1^T - \SCOVdv_1'^T, \SCOVdv_2^T, -\SCOVdv_2'^T]^T = 0$.
However, this is not possible since it implies that 
$\SCOVdv_2'=\SCOVdv_2=0$,
i.e., the sources corresponding to the DOAs $\DOAs_2$ and $\DOAs_2'$ have zero power.
\end{enumerate}
Thus, (\ref{eq:id-proof-2})
can only be satisfied in case 1) 
which proves the theorem. 

\end{proof}

\bibliographystyle{plain}

\end{document}